\documentclass[11pt,letterpaper]{article}

\usepackage[letterpaper,margin=1in]{geometry}
\usepackage{setspace}
\usepackage{tocloft}
\usepackage[utf8]{inputenc}
\usepackage[T1]{fontenc}
\usepackage{lmodern}
\usepackage{microtype}

\usepackage{amsmath,amssymb,amsfonts,amsthm,mathtools}
\usepackage{mathrsfs,bm}
\usepackage[dvipsnames,svgnames,x11names]{xcolor}
\usepackage{framed,graphicx,url}
\usepackage[lambda,advantage,operators,sets,adversary,landau,probability,notions,logic,ff,mm,primitives,events,complexity,asymptotics,keys]{cryptocode}

\usepackage{aliascnt}
\usepackage[bookmarks=true,pdfstartview=FitH,colorlinks,linkcolor=blue,filecolor=black,citecolor=Brown,urlcolor=blue]{hyperref}
\usepackage{cleveref}
\newtheorem{theorem}{Theorem}[section]
\newaliascnt{definition}{theorem}
\newtheorem{definition}[definition]{Definition}
\aliascntresetthe{definition}
\crefname{definition}{definition}{definitions}
\Crefname{definition}{Definition}{Definitions}
\newaliascnt{lemma}{theorem}
\newtheorem{lemma}[lemma]{Lemma}
\aliascntresetthe{lemma}
\crefname{lemma}{lemma}{lemmas}
\Crefname{lemma}{Lemma}{Lemmas}
\newaliascnt{corollary}{theorem}
\newtheorem{corollary}[corollary]{Corollary}
\aliascntresetthe{corollary}
\crefname{corollary}{corollary}{corollaries}
\Crefname{corollary}{Corollary}{Corollaries}
\newaliascnt{construction}{theorem}
\newtheorem{construction}[construction]{Construction}
\aliascntresetthe{construction}
\crefname{construction}{construction}{constructions}
\Crefname{construction}{Construction}{Constructions}
\newaliascnt{remark}{theorem}
\newtheorem{remark}[remark]{Remark}
\aliascntresetthe{remark}
\crefname{remark}{remark}{remarks}
\Crefname{remark}{Remark}{Remarks}
\newaliascnt{claim}{theorem}
\newtheorem{claim}[claim]{Claim}
\aliascntresetthe{claim}
\crefname{claim}{claim}{claims}
\Crefname{claim}{Claim}{Claims}
\crefname{equation}{}{}

\newcommand{\GapSVP}{\mathsf{GapSVP}}
\newcommand{\ff}[1]{\mathscr{#1}}
\newcommand{\ffe}[1]{\mathscr{#1}}
\newcommand{\game}[1]{\mathbf{#1}}
\newcommand{\Set}[1]{\mathcal{#1}}
\newcommand{\Dist}[1]{#1}
\newcommand{\att}[1]{\mathtt{#1}}
\newcommand{\Typ}{\mathsf{type}}
\newcommand{\Z}{\mathbb{Z}}
\newcommand{\N}{\mathbb{N}}
\newcommand{\PRG}{\mathsf{PRG}}
\newcommand{\BAD}{\mathcal{BAD}}
\newcommand{\TGGM}{\mathsf{TGGM}}
\newcommand{\Trunc}{\mathsf{Trunc}}
\newcommand{\CDH}{\mathsf{CDH}}
\newcommand{\AUX}{\mathcal{AUX}}
\newcommand{\skey}{\mat S}
\newcommand{\vcy}{\vc y}
\newcommand{\vcz}{\vc z}
\newcommand{\SD}{\mathsf{SD}}
\newcommand{\LWE}{\mathsf{LWE}}
\newcommand{\Ber}{\mathsf{Ber}}
\newcommand{\LPN}{\mathsf{LPN}}
\newcommand{\bina}{\{0, 1\}}
\renewcommand{\poly}{\mathsf{poly}}
\newcommand{\Adv}{\mathsf{Adv}}
\newcommand{\funceps}{\epsilon}
\newcommand{\func}[1]{{\MakeLowercase{#1}}}
\newcommand{\vc}[1]{{\mathbf{#1}}}
\newcommand{\mat}[1]{{\mathbf{#1}}}

\title{Pseudorandom Functions in $\mathrm{NC}^1$ from LWE/LPN/CDH\\(Or: How to Build Low-Depth PRFs, Generically)}

\author{Youlong Ding\thanks{The Hebrew University of Jerusalem, Israel. Email: \texttt{youlong.ding@mail.huji.ac.il}.}
\and Aayush Jain\thanks{Carnegie Mellon University, USA. Email: \texttt{aayushja@andrew.cmu.edu}.}
\and Ilan Komargodski\thanks{The Hebrew University of Jerusalem, Israel. Email: \texttt{ilank@cs.huji.ac.il}.}}

\date{}

\begin{document}
\hypersetup{pageanchor=false}

\maketitle

\begin{abstract}
We obtain new $\mathsf{NC}^1$-computable pseudorandom functions (PRFs) from various standard assumptions, resolving several long-standing open problems. Concretely, for the first time, we obtain $\mathsf{NC}^1$-computable PRFs:
\begin{itemize}
    \item 
    from the \textbf{Learning With Errors (LWE)} assumption with a polynomial modulus-to-noise ratio, improving upon prior low-depth constructions that required Ring-LWE with super-polynomial ratios [Banerjee-Peikert-Rosen, EUROCRYPT 2012];
    \item
    from the \textbf{Learning Parity with Noise (LPN)} assumption, removing the need for structured LPN variants [Boyle et al., FOCS 2020], [Ding-Jain-Komargodski, STOC 2025], [Ding-Jain-Komargodski, TCC 2026];
    \item 
    from the \textbf{Computational Diffie-Hellman (CDH)} assumption; prior works relied on the stronger Decisional Diffie-Hellman (DDH) or generalized Diffie-Hellman (GDH) assumptions [Naor-Reingold, FOCS '97, J. ACM '04].
    \end{itemize}

Technically, our results follow from a new generic transformation that converts weak PRFs computable in depth $d(n) = \Omega(\log n)$ to (strong) PRFs computable in depth $O(d(n))$. Roughly speaking, our approach combines \emph{tapering} the internal state of the classical GGM tree (so that the per-level depth decreases geometrically) with a two-stage security proof that transitions from a computational to an information-theoretic argument.
\end{abstract}






\thispagestyle{empty}

\newpage
\thispagestyle{empty}
\tableofcontents
\thispagestyle{empty}

\newpage
\setcounter{page}{1}
\hypersetup{pageanchor=true}
\section{Introduction}
Pseudorandom function families (PRFs)~\cite{JACM:GoldreichGM86}
are a fundamental primitive in cryptography, with applications
extending to computational complexity and learning theory.
Informally, a PRF is a keyed function
$F : \mathcal{K} \times \mathcal{X} \to \mathcal{Y}$
such that, for a randomly sampled key $k \in \mathcal{K}$,
no efficient adversary with oracle access can distinguish
$F(k,\cdot)$ from a randomly sampled function
$R : \mathcal{X} \to \mathcal{Y}$ with non-negligible advantage,
even when its queries are chosen adaptively.

The breadth of the applications makes the computational complexity
of PRFs a foundational question.
Of particular interest is \emph{circuit depth}, which captures
the parallel complexity of evaluation and affects both its own efficiency
and the scope of applications.
In cryptography, the depth of PRFs influences the round complexity
of secure computation protocols~\cite{C:DamgardI06,STOC:IshaiKOS08,EC:Benhamouda0KL21},
and $\mathsf{NC}^1$ PRFs\footnote{The class $\mathsf{NC}^1$ refers to polynomial-size,
logarithmic-depth Boolean circuits with bounded-fan-in gates.
Logarithmic depth is asymptotically optimal for PRFs in this circuit model.} enable advanced cryptographic
constructions~\cite{C:GorbunovVW12,EC:Zhandry25a}.
Beyond cryptography, low-depth PRFs have consequences for
computational complexity~\cite{CACM:Valiant84,STOC:KeaVal89,JCSS:RazborovR97,FOCS:ImpLev90,C:BlumFKL93},
including applications in meta-complexity~\cite{STOC:ApplebaumN25}.
They also yield conditional lower bounds for learning:
via the standard connection between pseudorandomness and learning,
secure PRFs in a circuit class $\mathcal{C}$ rule out efficient
learning of $\mathcal{C}$ under suitable access
models~\cite{JCSS:Kharitonov95,ToFC:BogdanovR17}.

Despite this motivation, our understanding of PRF circuit depth
lags behind that of circuit size.
For size, a strong generic result is known: any PRF family can be
transformed into one with linear-size evaluation
circuits~\cite{STOC:IshaiKOS08}
(see also~\cite{STOC:FanLY22}).
For depth, however, obtaining logarithmic-depth, i.e., $\mathsf{NC}^1$, PRFs is considerably more challenging.

\paragraph{Limitations of generic PRF constructions.}
The classic GGM construction~\cite{JACM:GoldreichGM86} built PRFs from any (length-doubling) pseudorandom generator (PRG).
However, this transformation is sequential and incurs a substantial depth overhead\footnote{One might hope to
compensate for the depth overhead by starting from a constant-depth PRG. See \Cref{sec:technicaloverview} for more discussion.}.
A more depth-efficient generic construction was later given by~\cite{FOCS:NaorR95}, who constructed PRFs from \emph{synthesizers}, a notion intermediate between PRGs and weak PRFs (WPRFs)\footnote{For WPRFs, the adversary only sees samples
$(x_i,F(k,x_i))$ on independently and uniformly sampled inputs
$x_i$, rather than being able to choose its queries. See \Cref{sec:CDH} for a detailed discussion of the differences between synthesizers and WPRFs.}.
Nevertheless, it still incurs a superconstant multiplicative depth overhead.
Consequently, an $\mathsf{NC}^1$ synthesizer does not yield $\mathsf{NC}^1$ PRFs via the NR construction.
These two constructions have long served as the main generic
constructions for building PRFs, yet neither provides a
depth-preserving transformation. Therefore, obtaining $\mathsf{NC}^1$
PRFs through these paradigms has remained elusive.

\paragraph{$\mathsf{NC}^1$ PRFs via direct constructions?}
This unsatisfactory state of affairs has motivated a
substantial line of work on \emph{direct} constructions of
low-depth PRFs from specific intractability assumptions.
The seminal work of Naor and
Reingold~\cite{FOCS:NaorR97,JACM:NaorR04} established the feasibility
of PRFs in $\mathsf{TC}^0 \subseteq \mathsf{NC}^1$ under the
decisional Diffie--Hellman ($\mathsf{DDH}$) and factoring assumptions.%
\footnote{The class $\mathsf{TC}^0$ refers to polynomial-size, constant-depth Boolean circuits with
unbounded-fan-in AND, OR, and majority gates, together with NOT
gates.}
Lewko and Waters~\cite{CCS:LewkoW09} generalized the Naor--Reingold
approach to the decisional $k$-Linear assumption.
Banerjee, Peikert, and Rosen~\cite{EC:BanerjeePR12} constructed
PRFs in $\mathsf{TC}^0$ from $\mathsf{RingLWE}$, albeit with a
super-polynomial modulus-to-noise ratio.
Applebaum and Raykov~\cite{TCC:ApplebaumR16a} constructed
PRFs in $\mathsf{TC}^0$ from a strong variant of
Goldreich's one-way function assumption.
More recently, Ding, Jain, and
Komargodski~\cite{STOC:DingJK25,TCC:DingJK26} constructed
$\mathsf{NC}^1$-computable PRFs from structured variants of $\LPN$,
including $\mathsf{SeededLPN}$, $\mathsf{RingLPN}$, and $\mathsf{SparseLPN}$ with respect to
low-depth expander graphs.
Complementing these constructions from existing assumptions,
Miles and Viola~\cite{C:MilesV12} proposed
substitution--permutation-network-inspired PRF candidates in $\mathsf{TC}^0$.
More recently, Ball et al.~\cite{TCC:BallDEKRS26}
proposed the first candidate PRF in $\mathsf{AC}^0[2] \subseteq \mathsf{NC}^1$.\footnote{The class $\mathsf{AC}^0[2]$ refers to polynomial-size,
constant-depth Boolean circuits with unbounded-fan-in
AND, OR, and parity gates, together with NOT gates.}

While these works have established the feasibility of
$\mathsf{NC}^1$ PRFs and broadened the range of assumptions
supporting them, they rely on ad hoc techniques tailored to the
underlying (typically structured) assumptions. Despite this effort, relatively few PRFs in
$\mathsf{NC}^1$ from standard or otherwise well-studied assumptions
were known.

\paragraph{The hash-then-evaluate heuristic?}
Given the limited range of known $\mathsf{NC}^1$ constructions,
a tempting alternative is to start from low-depth WPRFs and attempt to upgrade them to (strong) PRFs via the folklore
\emph{hash-then-evaluate}
heuristic~\cite{TCC:BonehIPSW18,FOCS:BoyleCGIKS20,STOC:DingJK25,CCDS:BrzuskaCEKM25}:
\[
    \mathsf{PRF}_k(x) \coloneqq \mathsf{WPRF}_k(H(x)).
\]
This simple
composition preserves the depth of the underlying WPRFs (with a low-depth hash function $H$).
The challenge is to justify its security.
Existing instantiations either rely
on additional assumptions beyond the security of the underlying
WPRFs or remain heuristic in the standard model.

\subsection{Our Contributions}
\paragraph{Overview.}
Perhaps surprisingly, our central contribution is a new
\emph{generic} construction that builds PRFs from WPRFs while preserving evaluation circuit depth.\footnote{More generally, our transformation applies to
synthesizers~\cite{FOCS:NaorR95,JCSS:NaorR99}, a weaker primitive
than WPRFs.}
In particular, $\mathsf{NC}^1$ WPRFs imply $\mathsf{NC}^1$ PRFs.
The transformation is provably secure (in the standard model).
This foundational advance has broad implications:
\begin{enumerate}
    \item It opens the door to new $\mathsf{NC}^1$ PRFs from a wide range of standard assumptions, resolving several long-standing open problems. Namely, by combining our transformation with new and existing
constructions of $\mathsf{NC}^1$ WPRFs and synthesizers, we obtain
$\mathsf{NC}^1$ PRFs from standard $\LWE$~\cite{JACM:Regev09}, standard $\LPN$~\cite{C:BlumFKL93} and $\CDH$~\cite{TIT:DiffieH76}.
    \item It also provides a \emph{provably secure alternative}
to the hash-then-evaluate heuristic.
In particular, existing $\mathsf{NC}^1$ WPRF candidates can now be upgraded to $\mathsf{NC}^1$ PRFs, under the same underlying
assumptions.
\end{enumerate}

Below, we elaborate on our results.

\paragraph{WPRFs $\Longrightarrow$ PRFs: Depth-preserving and provably secure.}
We show that WPRFs of depth $d(\lambda)=\Omega(\log\lambda)$ can be generically transformed to PRFs of depth
$O(d(\lambda))$, provided that $d(\lambda)/\log\lambda$ is non-decreasing.
The resulting PRFs are polynomially secure as long as the underlying WPRFs are polynomially secure.

\begin{theorem}[Informal; see Section~\ref{sec:taperingGGM}]
\label{thm:intro:weak2strong}
    If there exists a family of  WPRFs computable by Boolean circuits of depth $d(\lambda) = \Omega(\log \lambda)$ such that $d(\lambda)/\log\lambda$ is non-decreasing, then there exists a family of (strong) PRFs computable in depth $O(d(\lambda))$. In particular, WPRFs in $\mathsf{NC}^1$ imply PRFs in $\mathsf{NC}^1$.
\end{theorem}

\paragraph{Corollaries.}

We instantiate the above theorem with new and existing constructions of $\mathsf{NC}^1$ WPRFs (or synthesizers) to yield new constructions of $\mathsf{NC}^1$ PRFs, resolving several long-standing open problems:

\begin{enumerate}
\item \textbf{LWE (Section~\ref{sec:LWEPRF}).} We build the first  WPRFs in $\mathsf{NC}^1$ from standard $\LWE$ (with polynomial modulus-to-noise ratio, namely, $1/\alpha = n^{\epsilon}$ for any $\epsilon > 0$). Together with the theorem above, this yields the first PRFs in $\mathsf{NC}^1$ from $\LWE$ (with polynomial modulus-to-noise ratio) (and therefore, from the hardness of $\GapSVP$ lattice problems against polynomial-size quantum circuits with polynomial approximation ratio $\gamma = \tilde{O}(n^{1+\epsilon})$). Prior $\mathsf{NC}^1$ PRF constructions were based on $\mathsf{RingLWE}$ and required a super-polynomial ratio~\cite{EC:BanerjeePR12}. This result resolves a long-standing open problem, explicitly stated in Bogdanov-Rosen's survey \cite[Section 4.2]{ToFC:BogdanovR17} and Vaikuntanathan's lecture notes~\cite[Open Problem 5.1]{VinodLectureNotes}.

\item \textbf{LPN (Section~\ref{sec:LPNPRF}).}
We obtain the first $\mathsf{NC}^1$ PRFs from standard $\LPN$ with arbitrary inverse polynomial noise probability. This leverages the recent construction of  WPRFs in $\mathsf{NC}^1$ from the same assumption~\cite{STOC:DingJK25}. For code-based assumptions, the only $\mathsf{NC}^1$ construction from well-studied assumptions was known from structured variants, e.g., $\mathsf{RingLPN}$ or $\mathsf{SparseLPN}$ with respect to low-depth expander graphs by \cite{TCC:DingJK26}.

\item \textbf{CDH (Section~\ref{sec:CDHPRF}).}
We obtain the first $\mathsf{NC}^1$ PRFs from the Computational Diffie--Hellman ($\CDH$) assumption, improving upon the nearly thirty-year-old work of~\cite{FOCS:NaorR97,JACM:NaorR04}.
Our construction instantiates our transformation using a synthesizer (a weaker primitive than WPRFs) constructed from $\CDH$ in \cite{FOCS:NaorR95,JCSS:NaorR99}, together with the observation that our transformation, in fact, also applies to synthesizers.

\item From existing WPRF candidates. We can upgrade existing log-depth WPRF candidates (e.g., \cite{ITCS:AkaviaBGKR14,TCC:ApplebaumR16a,TCC:BonehIPSW18,C:BoyleCGIKS21,CRYPTO:ChenJL26}) to log-depth PRFs under the same underlying assumption without resorting to the hash-then-evaluate heuristic.

\end{enumerate}

\subsection{More Related Work}
We have discussed the most directly related work above.
Here, we review additional lines of research relevant to our results.

\paragraph{More constructions of PRFs.}
The existence of PRFs from arbitrary one-way functions follows
by combining the PRG construction of H{\aa}stad, Impagliazzo,
Levin, and Luby~\cite{SICOMP:HastadILL99} with
GGM~\cite{JACM:GoldreichGM86}.
Yu and Steinberger~\cite{EC:YuS16} constructed  PRFs
from $\LPN$ in
$\mathsf{AC}^{\epsilon}[2] \subseteq \mathsf{NC}^{1+\epsilon}$
for every constant $\epsilon \in (0, 1)$.
Li, Liu, Han, and Gu~\cite{DCC:LiLHG21} constructed PRFs
in $\mathsf{NC}^{1+\epsilon}$ for any constant
$\epsilon\in(0,1)$ based on standard $\LWE$ with polynomial
modulus.
Liu and Pass~\cite{EC:LiuP24} gave a direct construction of
quasi-polynomially secure PRFs from mild average-case hardness
of time-bounded Kolmogorov complexity.

\paragraph{Generic transformations for PRFs.}
Levin's domain extension technique \cite{STOC:Levin85,Comb:Levin87} extended any
superpolynomial-size domain PRFs to polynomial-length inputs while
preserving polynomial security.
Berman, Haitner, Komargodski, and
Naor~\cite{TCC:BermanHKN13,JOC:BermanHKN19} used cuckoo hashing
to obtain stronger concrete security guarantees, going beyond
the birthday bound of simple hashing.
D\"ottling and Schr\"oder~\cite{C:DottlingS15} developed
domain-extension and on-the-fly adaptation techniques that
convert suitable families of bounded-query PRFs into full-fledged PRFs.
Jain, Pietrzak, and Tentes~\cite{TCC:JainPT12} gave a
hardness-preserving construction of PRFs from length-doubling
PRGs that uses $O(\log q)$ PRG invocations per evaluation,
given an a priori bound $q$ on the number of adversarial queries.
Berman and Haitner~\cite{TCC:BermanH12,JOC:BermanH15}
gave a direct transformation from non-adaptively secure PRFs
to adaptively secure PRFs.
Soni and Tessaro~\cite{SCN:SoniT20} established limitations on
black-box proofs for broad classes of one-call and parallel-call
transformations based only on polynomial-time non-adaptive security.

\paragraph{Lower bounds on generic constructions of PRFs.}
Earlier work of Miles and Viola~\cite{JOC:MilesV15}
ruled out restricted black-box constructions of PRFs that make a single
PRG call and output one of its output bits.
Filmus, Ishai, Kaplan, and Kindler~\cite{CCC:FilmusIKK20, CC:FilmusIKK24} studied limitations on reducing PRF evaluation complexity
through separate preprocessing of the key and input.
Under restrictions on the input-encoding length, they ruled
out $\mathsf{AC}^0$ evaluation of exponentially secure PRFs
even with such preprocessing.
Beimel, Malkin, and Mazor~\cite{C:BeimelMM24} proved structural
lower bounds on black-box constructions of PRFs from PRGs.
They considered a class of tree constructions generalizing
GGM and ruled out logarithmic-depth constructions within this
class.
More recently, Alon, Dinur, and
Venkitasubramaniam~\cite{C:AlonDV26} proved query-complexity
lower bounds for fully black-box constructions of PRFs from PRGs.
Their main result applies to constructions making non-adaptive
PRG calls and using query-bounded reductions, and holds even
for one-bit WPRFs.

Our construction is consistent with the black-box lower bounds
for constructing PRFs from PRGs~\cite{JOC:MilesV15,
C:BeimelMM24,C:AlonDV26}.
These results concern constructions from an arbitrary PRG,
whereas we start from a WPRF or synthesizer.
Specifically, our proof uses pseudorandomness of the evaluation
matrix $(F(k_i,a_j))_{i,j}$ on independently sampled keys
and a shared pool of random inputs, a guarantee not supplied
by ordinary PRG security.
Our construction also lies outside the one-call and parallel-call
classes studied by Soni and Tessaro~\cite{SCN:SoniT20}:
each tree path contains a superconstant sequence of dependent
evaluations, with each output serving as the key for the next
evaluation.


\section{Technical Overview}
\label{sec:technicaloverview}

The main technical tool developed in this paper is a depth-preserving generic construction that builds  PRFs from weak PRFs (WPRFs), as stated in Theorem~\ref{thm:intro:weak2strong}.
Our method is largely based on the classical paradigm by \cite{JACM:GoldreichGM86} (see also  \cite{FOCS:NaorR95}) for upgrading weaker pseudorandom objects to pseudorandom functions.

We begin by revisiting the classic GGM construction to identify
the obstacles to achieving logarithmic depth. Throughout this
overview, $\lambda$ denotes the security parameter.

\paragraph{Why naive GGM does not suffice.}
Recall the GGM construction~\cite{JACM:GoldreichGM86}.
Starting from a seed $s$ of length $n = \lambda$, evaluation follows a path through a tree by repeatedly applying a length-doubling PRG to the current state. We can view the PRG as outputting $2$ blocks of $n$ bits each:
\[
    G(s) = G_0(s) \| G_1(s).
\]
For an input $x=(x_1,\ldots,x_\ell)\in \{0, 1\}^\ell$, the
PRF can then be defined as
\begin{equation}
\label{eq:GGM}
        F_s(x)
    \coloneqq
    G_{x_\ell}\circ\cdots\circ G_{x_1}(s).
\end{equation}

Motivated by giving the GGM construction the best possible chance to obtain logarithmic-depth PRFs, suppose that
\begin{enumerate}
    \item (Constant-depth) $G$ can be evaluated by circuits of constant depth $c$.
    \item (Poly-stretch) $G$ outputs $\lambda$ blocks of $\lambda$ bits each. Then each level of the tree can process an entire input chunk of $\log\lambda$ bits.
\end{enumerate}

By PRF domain
extension techniques~\cite{STOC:Levin85,Comb:Levin87}, it suffices to construct a PRF on a
superpolynomial-size domain. Thus, we may
take the number of chunks $\ell=\omega(1)$ to grow arbitrarily
slowly, obtaining a superpolynomially large input domain $[\lambda]^\ell$.
For any $\ell\leq\log\lambda$, a tempting depth estimate for the resulting PRFs according to \Cref{eq:GGM} is therefore
\[
    c\cdot \ell = O(\log\lambda).
\]
The issue is that it does not account for the cost of selecting a branch $G_{x_i}$. Note that a circuit must have fixed wiring and is not programmable: we cannot choose the subcircuits
$G_{x_1},\ldots,G_{x_\ell}$ in the beginning and then
compose them at no additional cost. A typical implementation is to let
each level evaluate the entire $G$ and then select the output block indexed by
$x_i$. Since selection among $\lambda$ branches has depth
$\Theta(\log\lambda)$, and this selection is repeated along the
sequential evaluation path, therefore the resulting circuit
 has depth $\Theta(\ell\log\lambda) = \omega(\log \lambda)$ even when $G$ has constant depth.

It is therefore natural to use a \emph{random-access} PRG: rather
than computing the entire output and then selecting a block, we
want a single circuit that takes both the index and the seed as
inputs and directly computes
\[
    \widehat{G}(j,s) \coloneqq G_j(s).
\]
Now if this circuit has depth $d(\lambda)$, the tree can
be evaluated in depth $\ell\,d(\lambda)$, without a
separate selection stage at each level. In particular, if
$d(\lambda)=o(\log\lambda)$, we can choose
$\ell=\omega(1)$ slowly enough that the total depth remains
$O(\log\lambda)$.

\[
\boxed{
\begin{gathered}
    \text{\bfseries Current status}\\
    \text{polynomial-stretch, random-access PRGs of depth } o(\log\lambda)\\
    \Longrightarrow\quad
    \text{PRFs in }\mathsf{NC}^1
\end{gathered}
}
\]

We next relax these requirements step by step, ultimately obtaining
a transformation from logarithmic-depth WPRFs
to logarithmic-depth PRFs.

\paragraph{Idea 1: Tapering GGM.}
We first relax the depth requirement for PRG from $o(\log \lambda)$ to $O(\log \lambda)$.
Our idea is \underline{shrinking the internal state across levels} in the GGM tree. Therefore, their evaluation depths decrease
geometrically as we move down the tree.

Set the number of input chunks $\ell=\log\log\log\lambda$ and define the tapering parameters
\[
    \lambda^{(i)}
    \coloneqq
    \lambda^{2^{-i+1}},
    \qquad i\in[\ell+1].
\]
Thus, $\lambda^{(1)}=\lambda$, $\lambda^{(\ell+1)}=\lambda^{1/\log\log \lambda}$, and
$\lambda^{(i+1)}=\sqrt{\lambda^{(i)}}$.
At level $i$, we use a PRG
\begin{equation}
\label{eq:TGGM}
    G^{(i)}:
    \{0,1\}^{\lambda^{(i)}}
    \longrightarrow
    \bigl(\{0,1\}^{\lambda^{(i+1)}}\bigr)^\lambda.
\end{equation}
Note that the branching factor remains $\lambda$, while the state length
decreases from $\lambda^{(i)}$ to $\lambda^{(i+1)}$.

Since $G^{(i)}$ can be evaluated
in depth $d_G(\lambda^{(i)})$, the depth of this tapering GGM
construction is then
\[
    d_{\mathsf{TGGM}}(\lambda)
    \leq
    \sum_{i=1}^{\ell} d_G\bigl(\lambda^{(i)}\bigr).
\]
For $d(\lambda)=O(\log \lambda)$, we obtain
\[
    \sum_{i=1}^{\ell} d_G\bigl(\lambda^{(i)}\bigr)
    =
    O\!\left(
        \log\lambda \cdot
        \sum_{i=1}^{\ell}2^{-i+1}
    \right)
    =
    O(\log\lambda).
\]
This resolves the depth bottleneck, but comes at a cost in
security and stretch. Observe that near the bottom of the tree, the \underline{seed length shrinks to $\lambda^{o(1)}$}, bringing two additional requirements:
\begin{enumerate}
    \item \emph{Superpolynomial security.}
    An adversary running in polynomial time in the global
    parameter $\lambda$ runs in superpolynomial time in the
    local seed length $\lambda^{o(1)}$. Thus, a direct security proof
    requires the PRG to have superpolynomial security.

    \item \emph{Superpolynomial stretch.}
    The PRG must still output $\lambda$ blocks, one for each
    branch, so its output length remains $\Omega(\lambda)$.
    Since this is superpolynomial in the local seed length, we require the PRG to have superpolynomial stretch.
\end{enumerate}

\[
\boxed{
\begin{gathered}
    \text{\bfseries Current status}\\[3pt]
    \text{Superpoly secure, superpoly-stretch }
    \text{random-access PRGs of depth } O(\log \lambda)\\[3pt]
    \Longrightarrow\quad
    \text{(Poly secure) PRFs in }\mathsf{NC}^1
\end{gathered}
}
\]

\paragraph{Some syntactic sugar via WPRFs.}
Our ultimate goal is to rely only on a polynomially secure
underlying primitive. Before addressing this issue, we introduce some syntactic sugar to first remove the need
for an explicit superpolynomial-stretch and random access requirement.
The idea is that we can use WPRFs to simulate a superpoly-stretch random-access PRG via evaluation on presampled random inputs.

Suppose that we are given a WPRF family $W=\{W_{\lambda}\}_{\lambda\in\mathbb N}$, where
  $$W_{\lambda} \colon \{0,1\}^{\lambda} \times \{0,1\}^{\lambda} \to \{0,1\}^{\sqrt{\lambda}}.$$
  We regard the first argument as the key and the second as the
input, and assume that $W_{\lambda}$ can be computed in depth $d(\lambda) = O(\log \lambda)$. For the $j$-th block of the output of PRG, we define
\[
G_j(\cdot) \;\coloneqq\; W\bigl(\cdot,\, r_j\bigr),\qquad j\in[\lambda],
\]
where $r_j \sample \{0,1\}^{\lambda}$ is a presampled random input for WPRFs.

Now, in the construction, the secret key of the PRF consists of:
\begin{enumerate}
    \item A root key for the WPRF $\key_{\mathcal{W}}\sample \{0,1\}^{\lambda}$.
    \item Randomly chosen strings from the input domain of the WPRF:
$$
\key_{\mathcal{WX}} = (r^{(i)}_1,\ldots,r^{(i)}_{\lambda})_{i\in[\ell]},
$$
where each $r^{(i)}_j \sample \{0, 1\}^{\lambda^{(i)}}$. Recall that $\lambda^{(i)}$ is the tapering parameter defined before.
\end{enumerate}
Then, on input $x\in[\lambda]^{\ell}$, the (strong) PRF computes
\[
  y^{(1)}\gets \key_{\mathcal{W}},\qquad y^{(i+1)}\gets G^{(i)}_{x_i}(y^{(i)}) = W_{\lambda^{(i)}}(y^{(i)}, r^{(i)}_{x_i})\ \text{ for }\ i=1,\dots,\ell,
\]
and finally outputs $y^{(\ell+1)}$.

To compute this PRF in low-depth, we can first select all relevant $r^{(i)}_{x_i}$'s to be used in the tree,  requiring  $O(\log \lambda)$ depth. We then evaluate the tapering GGM tree by invoking the  WPRF $\ell$ times sequentially. Then the exact same computation yields $O(\log \lambda)$ total evaluation depth. Again, we still require the WPRFs to have super-poly security because near the bottom of the tree, it only has subpolynomial key length.

\[
\boxed{
\begin{gathered}
    \text{\bfseries Current status}\\[3pt]
    \text{Superpoly secure WPRFs in }\mathsf{NC}^1\\[3pt]
    \Longrightarrow\quad
    \text{(Poly secure) PRFs in }\mathsf{NC}^1
\end{gathered}
}
\]

We remark that this technical step mainly serves to clarify what remains to be
done. It is nevertheless helpful to keep the PRG perspective in mind, and we will use the terms PRG and WPRFs interchangeably in what follows.

\paragraph{Idea 2: A two-stage security proof.}

It remains to eliminate the superpolynomial security requirement
on the underlying WPRFs, thereby obtaining our main theorem.

Recall that the classic GGM security proof proceeds through a sequence of
hybrids, replacing the states at nodes reached by the adversary's
queries with independent uniform values, one level at a time.
Applying this argument throughout our tapering GGM tree would
eventually require invoking security at subpolynomial key lengths,
apparently forcing us to assume superpolynomial security of the underlying PRG/WPRFs.

Our key idea is that this computational security argument need not
continue all the way to the leaves. Instead, we can actually stop once the
query paths have separated. For any polynomial number of  queries,
we can make this separation occur within a constant
number of levels, where the local key lengths are still
polynomially related to $\lambda$. We then handle the remaining
levels information-theoretically, without invoking computational
security at their smaller key lengths.
This gives the two-stage proof illustrated
in \Cref{fig:ggm-proof-split}. Details follow.

We first explain the argument for non-adaptive PRF security, where
the adversary fixes all its queries before receiving any answers.
Consider an efficient adversary making $q=\poly(\lambda)$ queries
$x_{(1)},\ldots,x_{(q)}\in\{0,1\}^{\lambda}$. Without loss of generality, assume the queries are distinct.

\begin{figure}[h!]
    \centering
    \includegraphics[width=0.9\linewidth]{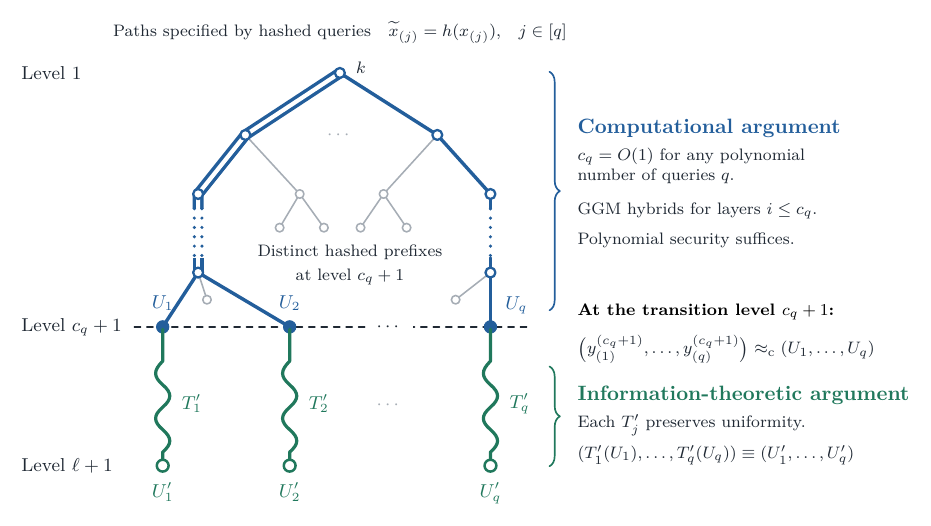}
    \caption{The two-stage security proof for the tapering GGM
    construction, illustrated.}
    \label{fig:ggm-proof-split}
\end{figure}

\paragraph{Stage I: Computational argument via GGM-style proof.}
To enforce that the queries separate within constant levels, we will first hash the queries using a pairwise-independent hash
function $h:\{0,1\}^{\lambda}\longrightarrow[\lambda]^\ell$,
whose seed is included in the PRF key.
Write
\[
    \widetilde{x}_{(j)}\coloneqq h(x_{(j)}),
    \qquad j\in[q].
\]
Let $(\star)$ denote the event that the length-$c_q$
prefixes of these hashed queries are pairwise distinct.
Equivalently, their paths reach distinct nodes at level $c_q+1$ and are therefore separated.
By pairwise independence and a union bound,
\[
    \Pr[\neg(\star)]
    \leq
    \binom{q}{2}\lambda^{-c_q}.
\]
Since $q$ is polynomial in $\lambda$, we can choose a sufficiently
large constant
$c_q=O(1+\log_\lambda q) = O(1)$ so that the event ($\star$) happens with probability at least 1/2.

Condition on event $(\star)$.
Via a GGM-style hybrid argument, the states at these
$q$ distinct nodes are jointly pseudorandom:
\begin{equation}
\label{eq:overviewinter}
    \bigl(y_{(1)}^{(c_q+1)},\ldots,y_{(q)}^{(c_q+1)}\bigr)
    \approx_{\mathrm{c}}
    (U_1,\ldots,U_q),
\end{equation}
where the $U_j$ are independent uniform strings of the
appropriate length.
Note that this argument invokes WPRF security only at levels $i\le c_q$, and at each such level, $\lambda^{(i)}=\lambda^{2^{-i+1}}$
is a polynomial in $\lambda$ (because $c_q$ is a constant), so polynomial security of the underlying PRGs/WPRFs suffices.

\begin{remark}
\normalfont
The number $c_q$ of computational hybrid steps is chosen in the security proof depending on the adversary's number of queries\footnote{
The formal proof starts by: Assume, for the sake of contradiction, that there exists an efficient adversary making $q$ queries that breaks our PRF
construction with non-negligible advantage. We then choose
$c_q$ depending on $q$ and then reach a contradiction to the polynomial security of underlying WPRFs via a hybrid argument.
}.
We note that the PRF construction itself does not
depend on $q$ or the adversary, and the construction is secure against any efficient adversary (making any polynomial number of queries).
\end{remark}

\paragraph{Stage II: Information-theoretic argument via
key-uniform WPRFs.}
So far we have only proven pseudorandomness at level $c_q+1$. Observe that the final output is computed by evaluating the function on $y^{(c_q+1)}_{(j)}$ (which is already pseudorandom) according to the remaining suffix of the hashed input of~$x_{(j)}$. Our idea is that the pseudorandomness at the internal level $c_q+1$ can be freely propagated to the final output level \emph{information-theoretically}.

To this end, we make the underlying WPRFs satisfy a property
that we call \emph{key-uniformity}: for every fixed input $r$,
the output $W(K,r)$ is exactly uniform when $K$ is a uniformly
random key.
This property can be easily obtained from any WPRF by simply extending the key with an
independent uniform pad $\rho$ of the output length and defining
\[
    \widehat{W}\bigl((k,\rho),r\bigr)
    \coloneqq
    W(k,r)\oplus\rho.
\]
This preserves WPRF security and makes the output exactly uniform. We use this key-uniform version for all WPRFs in the GGM tree. As a result, every induced transformation
\[
    T_r(\cdot)\coloneqq W(\cdot,r)
\]
is uniformity-preserving: it maps a uniformly random input
to a uniformly random output.

Then for $j \in [q]$, on input $x_{(j)}$, the output of the PRF  can be written as
$$
y_{(j)}^{(\ell + 1)} =  T'_j\left(y_{(j)}^{(c_q+1)}\right) \coloneqq T_j^{(\ell)}\circ \cdots \circ T_j^{(c_q+1)} \left(y_{(j)}^{(c_q+1)}\right),
$$
where $T_j^{(i)}$ denotes the uniformity-preserving transformation at layer $i$.
Since uniformity preservation is closed under composition, the map $T'_j$ is also uniformity-preserving.
It then follows that the PRF outputs are pseudorandom:
\begin{align*}
 \left(y_{(1)}^{(\ell + 1)}, \cdots, y_{(q)}^{(\ell + 1)}\right) = \left(T'_1\left(y_{(1)}^{(c_q+1)}\right), \dots, T'_q\left(y_{(q)}^{(c_q+1)}\right)\right) & \approx_{\mathrm{c}} (T'_1(U_1), \dots, T'_q(U_q)) \\  & \equiv  (U'_{1}, \dots, U'_{q}).
\end{align*}

\paragraph{Completing the proof.}
We have just proved non-adaptive PRF security conditioned on event $(\star)$. The following two final steps will complete the proof:
\begin{itemize}
    \item (Amplification for event ($\star$))
    We can simply evaluate $m$ independent copies of the construction and XOR the outputs. Whenever the event $(\star)$ happens for at least one copy, the final output is pseudorandom.

    \item (Adaptive security) One can use generic non-adaptive-to-adaptive
    transformations~\cite{TCC:BermanH12,JOC:BermanH15}
    to upgrade non-adaptive PRFs to adaptively
    secure PRFs (which preserves low depth).
    In fact,  we directly prove adaptive security of
    our construction, without applying
    an additional transformation.
\end{itemize}

\[
\boxed{
\begin{gathered}
    \text{\bfseries Goal achieved}\\[3pt]
    \text{(Poly secure) WPRFs in }\mathsf{NC}^1\\[3pt]
    \Longrightarrow\quad
    \text{(Poly secure) PRFs in }\mathsf{NC}^1
\end{gathered}
}
\]

\subsection{\texorpdfstring{Weak PRFs in $\mathsf{NC}^1$ from $\LWE$}{Weak PRFs in NC1 from LWE}}
We then instantiate our generic construction under several standard assumptions, including $\LPN$, $\LWE$ (with any polynomial modulus-to-noise ratio), and $\CDH$.
While the above discussion is phrased in terms of WPRFs, the transformation actually applies more generally to the weaker notion of synthesizers.\footnote{The reason is that, in our construction, the presampled inputs $r_j$ to the WPRF are incorporated into the PRF key and are not exposed to the adversary. Thus, the full weak pseudorandomness guarantee is unnecessary, and the weaker notion of synthesizers is sufficient.}
For $\LPN$, we use the recent result of~\cite{STOC:DingJK25}, which gives WPRFs in $\mathsf{NC}^1$ from standard $\LPN$.
For $\CDH$, we use this more general perspective together with the $\mathsf{NC}^1$ synthesizers of~\cite{FOCS:NaorR95,JCSS:NaorR99}.

For $\LWE$, we give the first construction of WPRFs in $\mathsf{NC}^1$ from polynomial-modulus $\LWE$, by combining and optimizing several known ideas. In particular, we make use of the chaining idea proposed by \cite{EC:Kim20} and further developed in \cite{STOC:DingJK25}.
In this framework, the problem of constructing low-depth WPRFs essentially reduces to designing a sublogarithmic-depth (i.e., $o(\log \lambda)$) sampler for the LWE error distribution.
We show how to sample the rounded Gaussian error distribution~\cite{STOC:Regev05} in such low depth.\footnote{We remark that the original formulation of LWE in~\cite{STOC:Regev05} uses continuous Gaussian errors, although in current usage, LWE is more commonly instantiated with discrete Gaussian errors. See \Cref{sec:LWEPRF} for a detailed discussion.}
In more detail, we implement the Box--Muller transform at finite precision to approximate a Gaussian sample with standard deviation $\sigma=\alpha q$, and then round to the nearest integer modulo $q$. The resulting distribution is within statistical distance $2^{-\Omega(\kappa)}$ of the rounded Gaussian when $\log q=o(\kappa)$, and the circuit has depth $O(\poly(\log \kappa))$. For polynomially bounded $q$, setting $\kappa=\log^2\lambda$ gives negligible statistical error and sublogarithmic depth.

Consequently, we obtain the following theorem and refer to \Cref{sec:PRFfromLWE} for more details.
\begin{theorem}[Informal; see Section~\ref{sec:LWEPRF}]
There exist WPRFs in $\mathsf{NC}^1$ if there exists $\epsilon > 0$ such that the $\LWE$ assumption holds for modulus-to-noise ratio $1/\alpha = n^{\epsilon}$.
\end{theorem}
\section{Preliminaries}

\paragraph{Notation.}
We use $[n]$ to denote the set $\{1,\ldots,n\}$.
For finite sets $\Set X$ and $\Set Y$, we write $\ff U(\Set X,\Set Y)$ for the set of all functions from $\Set X$ to $\Set Y$.
For a probability distribution $\Dist X$ over a domain $\Set D$, let $\Dist X^n$ denote the $n$-fold product distribution over $\Set D^n$.
The uniform distribution over a finite domain $\Set D$ is denoted by $U(\Set D)$.
We write $X\sim \Dist D$ to denote that the random variable $X$ is distributed according to $\Dist D$.
For a distribution $\Dist D$, we write $s\sample \Dist D$ to sample $s$ according to $\Dist D$; for a finite set $\Set S$, we write $s\sample \Set S$ to sample $s$ uniformly from $\Set S$.
For two distributions $X$ and $Y$ over the same finite domain, $\SD(X,Y)$ denotes their statistical distance.
Throughout the paper, whenever a quantity such as $\log\lambda$ or $\lambda^{2^{-1}}$ is required to be an integer, it is rounded up to the nearest integer.
We use the following vector convention: a vector is treated as a row vector when multiplied on the left and as a column vector when multiplied on the right.
\paragraph{Circuit model.}
We work in the Boolean circuit model with bounded fan-in $\{\mathrm{AND},\mathrm{OR},\mathrm{NOT}\}$ gates.

\begin{definition}[$(s, d)$-Function Families]
\label{def:functionfamily}
\normalfont
A \emph{function family} $\ff F$ with key space $\Set K$ is a collection
\[
\ff F=\{F_k\colon \Set X\to \Set Y,\; k\in\Set K\}\subseteq \ff U(\Set X,\Set Y).
\]
Let $s,d\in\N$.
The family $\ff F$ is an $(s,d)$-\emph{function family} if the universal evaluation function $F\colon\Set K\times\Set X\to\Set Y$ is computable by a Boolean circuit of size at most $s$ and depth at most $d$.
\end{definition}

\begin{definition}[$(q, t, \epsilon)$-Indistinguishable]
\label{def:indistinguishable}
\normalfont
We model an adversary by a circuit equipped with oracle gates interacting with probabilistic experiments, also called \emph{games}.
For an adversary $\att A$ and two games $\game G_0,  \game G_1$, define its \emph{distinguishing advantage} as
\[
\Adv_{\att A}(\game G_0,\game G_1)
\;\coloneqq\;
\bigl|\Pr[\att A \text{ outputs }1 \text{ in } \game G_0]-\Pr[\att A \text{ outputs }1 \text{ in } \game G_1]\bigr|.
\]

For $q,t\in\N$ and $\epsilon > 0$, the games $\game G_0$ and $\game G_1$ are $(q,t,\epsilon)$-\emph{indistinguishable}, denoted $\game G_0\approx_{q,t,\epsilon}\game G_1$, if every size-$t$ adversary making at most $q$ oracle queries has advantage less than $\epsilon$.
For games that give the adversary a sample and no oracle access, we use $q=0$.
We write $\game G_0\approx_{q,\infty,0}\game G_1$ for perfect $q$-query indistinguishability, namely identical views for every adversary making at most $q$ oracle queries.
\end{definition}

\begin{lemma}[Triangle Inequality]
\label{lemma:triangle}
If $\game H_0,\game H_1$ are $(q_1,t_1,\epsilon_1)$-indistinguishable and $\game H_1,\game H_2$ are $(q_2,t_2,\epsilon_2)$-indistinguishable, then $\game H_0,\game H_2$ are $(\min\{q_1,q_2\},\min\{t_1,t_2\},\epsilon_1+\epsilon_2)$-indistinguishable.
\end{lemma}

The next two definitions describe PRF and WPRF security.
Strong pseudorandomness is the usual PRF experiment: the adversary chooses its oracle queries, possibly adaptively.
Weak pseudorandomness gives the adversary only random input-output pairs, which is the security notion that our generic construction will build on.
\begin{definition}[$(q, t, \epsilon)$-Strong Pseudorandom]
\label{def:pseudorandom}
\normalfont
Let $q, t \in \N$, $\epsilon \geq 0$.
A function family $\ff F=\{F_k\colon\Set X\to\Set Y,\; k\in\Set K\}$ is \emph{$(q,t,\epsilon)$-strong pseudorandom} if $\game G_0\approx_{q,t,\epsilon}\game G_1$, where:
\begin{itemize}
    \item \textbf{Game $\game  G_0$.} Sample a random key $k\sample \Set{K}$, and give the adversary oracle access to $F_k(\cdot)$.
    \item \textbf{Game $\game  G_1$.} Sample a  random function $R \sample \ff U(\Set X, \Set Y)$, and give the adversary oracle access to $R(\cdot)$.
\end{itemize}
\end{definition}

We also define weak pseudorandomness, where the adversary does not obtain oracle access, but only receives input-output pairs on uniformly random inputs.
This yields the notion of weak pseudorandom functions (WPRFs), formally defined below.

\begin{definition}[$(q, t, \epsilon)$-Weak Pseudorandom]
\label{def:weakpseudorandom}
\normalfont
Let $q, t \in \N$, $\epsilon \geq 0$.
A function family $\ff G=\{G_k\colon \Set{X}\to \Set{Y}\}_{k\in\Set{K}}$ is \emph{$(q, t,\epsilon)$-weak pseudorandom} if $\game G_0 \approx_{0, t, \epsilon} \game G_1$:
\begin{itemize}

    \item \textbf{Game $\game G_0$.} Sample random inputs $\vc x=(x_{(1)},\ldots,x_{(q)})\sample\Set X^q$ and a random key $k\sample\Set K$, and give the adversary
    \[
    (\vc x,G(k,\vc x))\in\Set X^q\times\Set Y^q,
    \qquad
    G(k,\vc x)\coloneqq (G(k,x_{(1)}),\ldots,G(k,x_{(q)})).
    \]
    \item \textbf{Game $\game G_1$.} Sample $\vc x\sample\Set X^q$ and $\vc u=(u_{(1)},\ldots,u_{(q)})\sample\Set Y^q$, and give the adversary $(\vc x,\vc u)$.
\end{itemize}
\end{definition}

We will also use the notion of $n$-wise independent hashing.
For our purposes, it is convenient to phrase $n$-wise independent hashing as a special case of perfect strong pseudorandomness against $n$ queries.
\begin{definition}[$n$-Wise Independent Hashing]
\label{def:kwise}
\normalfont
Let $n \in \N$.
A {function family} $\ff H$ is an \emph{$n$-wise independent hashing} if $\ff H$ is $(n, \infty, 0)$-strong pseudorandom. Equivalently, on every set of at most~$n$ distinct inputs, the outputs are independent and uniform.
\end{definition}

A basic consequence of $2$-wise independence is that collisions are unlikely on any fixed set of distinct inputs.
\begin{lemma}[Collision Probability]
\label{lemma:collision}
Let $\ff H = \{H_k\colon \Set X \to \Set Y, \, k\in \Set K\}$ be a $2$-wise independent hashing.
For any distinct $x_{(1)}, \dots, x_{(q)} \in \Set X$, it holds that
$$
\Pr_{k \sample \Set K}\left[H_k\left(x_{(1)}\right), \dots, H_k\left(x_{(q)}\right) \text{ not all distinct}\right] \leq \frac{q^2}{2} \cdot \frac{1}{|\Set Y|}
$$
\end{lemma}
\begin{proof}
For every pair of distinct inputs, $2$-wise independence makes the two hash values independent and uniform, so they collide with probability $1/|\Set Y|$.
A union bound over the $\le q^2/2$ pairs gives the claim.
\end{proof}

We will repeatedly combine function families in two simple ways.
The sum operation XORs independent outputs and is used to amplify the randomness guarantee.
The product operation concatenates independent outputs and is used to extend the range of pseudorandom objects.

\begin{definition}[Sum of Function Families]
\normalfont
\label{def:sumoffunctionfamily}
Let $\Set Y = \bina^{b}$ for a fixed output length $b\in\N$, and
let $n\in \N$.
For $i \in [n]$,
let $\ff F_i = \{F_k\colon \Set X \to \Set Y, k \in \Set K_i\}$ be a $(s_i, d_i)$-function family.
The \emph{sum} of these families is
\[
\ff F_1\oplus \cdots \oplus \ff F_n \;\coloneqq\; \{\, F_{k_1}\oplus \cdots \oplus  F_{k_n} \colon \Set{X}\to \Set Y,\;  (k_1, \cdots,  k_n) \in \Set K_1 \times \cdots \times \Set K_n \,\},
\]
where $$(F_{k_1}\oplus \cdots \oplus F_{k_n})(x)\coloneqq F_{k_1}(x)\oplus \cdots \oplus F_{k_n}(x) \in \Set Y.$$

It is a $(\sum_i s_i+O(nb),\max_i d_i+O(\log n))$-function family.
Let $\mathcal{F}^{\oplus n}$
denote the $n$-fold direct sum of $\mathcal{F}$.
\end{definition}

The product of function families corresponds to evaluating all component functions on the same input $x$, and then concatenating their outputs.
\begin{definition}[Product of Function Families]
\label{def:concatoffunctionfamily}
\normalfont
Let $n\in \N$ and let $\ff F  = \{ F_k \colon \Set{X}\to \mathcal{Y}, \, k\in \Set K\}$ be an $(s, d)$-function family.
The $n$-product of $\ff F$ is
\[
\ff F^{\times n}
\;\coloneqq\;
\{\, (F_{k_1}, \dots,  F_{k_n}) \colon \Set X \to \mathcal{Y}^n,\;  (k_1, \dots, k_n) \in \Set K^n \,\},
\]
where $$(F_{k_1}, \dots, F_{k_n})(x) \coloneqq (F_{k_1}(x), \dots, F_{k_n}(x)) \in \Set Y^n.$$

It is an $(ns,d)$-function family.
\end{definition}

For the following closure properties, let
\[
\rho(q,n,s,a,b)\coloneqq Cnq^2(s+a+b+1),
\]
where $C>0$ is a sufficiently large universal constant.

The first closure property says that adding an independent pseudorandom component preserves pseudorandomness.
Intuitively, as long as one summand already looks random, the whole sum also looks random.
\begin{lemma}[Direct Sum Lemma]
\label{lemma:superposition}
Let $\Typ \in \{\text{strong}, \text{weak}\}$.
Suppose that each $\ff F_i\colon\bina^a\to\bina^b$ has evaluation size at most $s$.
If some $\ff F_i$ is $(q,t,\epsilon)$-$\Typ$ pseudorandom, then
$\ff F_1\oplus\cdots\oplus\ff F_n$ is
$(q,t-\rho(q,n,s,a,b),\epsilon)$-$\Typ$ pseudorandom, whenever the stated circuit-size bound is positive.
\end{lemma}

The second closure property lets us increase the output length while paying only a linear loss in security.
\begin{lemma}[Product Preserves Pseudorandomness]
\label{lemma:productpreserve}
Let $\Typ \in \{\text{strong}, \text{weak}\}$.
Let $a$ and $b$ be the binary encoding lengths of $\Set X$ and $\Set Y$.
If the $(s,d)$-function family $\ff F\colon\Set X\to\Set Y$ is $(q,t,\epsilon)$-$\Typ$ pseudorandom,
then $\ff F^{\times n}\colon\Set X\to\Set Y^n$ is an $(ns,d)$-function family that is
$(q,t-\rho(q,n,s,a,b),n\epsilon)$-$\Typ$ pseudorandom, whenever the stated circuit-size bound is positive.
\end{lemma}

Conversely, pseudorandomness is also preserved when we discard part of the output.
\begin{lemma}[Output Restriction Preserves Pseudorandomness]
\label{lemma:RangeRestrictionpreservepseudorandom}
Let $\Typ \in \{\text{strong}, \text{weak}\}$.
Let  $\ff F = \{ F_k \colon \mathcal{X} \to \mathcal{Y}^n\}$ be $(q, t, \epsilon)$-$\Typ$ pseudorandom.
Then for any $m' \leq n$, $\ff F' = \{ F'_k \colon \mathcal{X} \to \mathcal{Y}^{m'}\}$ is $(q, t, \epsilon)$-$\Typ$ pseudorandom, where $F'_k(x) \coloneqq (F_k(x))_{1\dots m'} \in \Set Y^{m'}$.
\end{lemma}

\begin{corollary}[Range Extension]
\label{corollary:rangeextension}
Let $\Typ \in \{\text{strong}, \text{weak}\}$.
If there exists an $(s, d)$-function family $\ff F = \{F_k \colon \Set X \to \Set Y^m, k \in \Set K\}$ that is $(q, t, \epsilon)$-$\Typ$ pseudorandom,
then there exists an $(ns,d)$-function family $\ff F' = \{F_k\colon\Set X\to\Set Y^n, k\in\Set K^n\}$ that is $(q,t-\rho(q,n,s,a,mb),n\epsilon)$-$\Typ$ pseudorandom, whenever the stated circuit-size bound is positive.
Here $a$ and $b$ are the binary encoding lengths of $\Set X$ and $\Set Y$, respectively.
\end{corollary}

\subsection{Pseudorandom Functions}
\label{sec:prf-definitions}

In this section, we formally define pseudorandom functions.

\begin{definition}[Polynomially Bounded Functions]
\normalfont
A function $f (\cdot)$ is \emph{polynomially bounded} if there exists a polynomial $p(\cdot)$, such that $f(n) \leq p(n)$, for $n \in \N$.
\end{definition}

\begin{definition}[Negligible Functions]
\normalfont
A function $\funceps (\cdot)$ is \emph{negligible} if for every $c>0$, there exists $n_c\in\N$ such that $\funceps(n) < n^{-c}$ for all $n \ge n_c$.
\end{definition}

\begin{definition}[Function Family Ensembles]
\label{def:ffe}
\normalfont
Let $p_0(\cdot),p_1(\cdot),p_2(\cdot)$ be positive integer-valued, polynomially bounded functions.
A \emph{function family ensemble} $\ffe F = \{\ff F_{\lambda}\}_{\lambda \in \N}$ is an infinite set of function families, where
$$\ff F_{\lambda} = \{ F_k \colon \bina^{p_0(\lambda)} \to \bina^{p_1(\lambda)}, \, k \in \bina^{p_2(\lambda)}\}.$$
The usual input-length parameterization is the special case $p_0(\lambda)=\lambda$.

Let $t(\cdot), d(\cdot)$ be functions of $\lambda$.
$\ffe{F} = \{\ff F_{\lambda}\}_{\lambda \in \N}$ is a $(t, d)$-\emph{function family ensemble} if for every $\lambda \in \N$, $\ff F_{\lambda}$ is a $(t(\lambda), d(\lambda))$-function family.
$\ffe{F}$ is \emph{efficient} if there exist two polynomials $t(\cdot)$ and $d(\cdot)$ such that $\ffe{F}$ is a $(t, d)$-function family ensemble.
Moreover, $\ffe{F}$ is in $\mathsf{NC}^1$ if it is a $(\poly(\lambda), O(\log \lambda))$-function family ensemble.
\end{definition}

\begin{definition}[PRFs~\cite{FOCS:GoldreichGM84,JACM:GoldreichGM86}]
\label{def:prf}
\normalfont
Let $\func q(\cdot)$, $\func t(\cdot)$, and $\funceps (\cdot)$ be functions of the security parameter $\lambda$.
An efficient function family ensemble $\ffe{F}=\{\ff F_\lambda \}_{\lambda\in\N}$ is a \emph{$(\func q, \func t, \funceps)$-PRF} if for every $\lambda\in\N$, $\ff F_\lambda$ is $(q(\lambda), t(\lambda), \funceps(\lambda))$-strong pseudorandom.

$\ffe{F}$ is a (\emph{polynomially secure}) \emph{PRF}  if for any polynomials $\func q(\cdot)$, $\func t(\cdot)$, there exists a negligible function $\funceps(\cdot)$, such that $\ffe{F}$ is a $(\func q, \func t, \funceps)$-\emph{PRF}.

\end{definition}

WPRFs are defined analogously to PRFs, except that the definition is with respect to weak pseudorandomness rather than strong pseudorandomness.

\begin{definition}[WPRFs~\cite{FOCS:NaorR95,JCSS:NaorR99}]
\label{def:weakprf}
\normalfont

Let $\func q(\cdot)$, $\func t(\cdot)$, and $\funceps (\cdot)$ be functions of the security parameter $\lambda$.
An efficient function family ensemble $\ffe{G}=\{\ff G_\lambda \}_{\lambda\in\N}$ is a \emph{$(\func q, \func t, \funceps)$-WPRF} if for every $\lambda\in\N$, $\ff G_\lambda$ is $(q(\lambda), t(\lambda), \funceps(\lambda))$-weak pseudorandom.

$\ffe{G}$ is a (\emph{polynomially secure}) \emph{WPRF}  if for any polynomials $\func q(\cdot)$, $\func t(\cdot)$, there exists a negligible function $\funceps(\cdot)$, such that $\ffe{G}$ is a $(\func q, \func t, \funceps)$-\emph{WPRF}.
\end{definition}

\subsection{From Non-Adaptive to Full-Fledged PRFs}
\label{sec:almostrandom}

Our generic construction in \Cref{sec:taperingGGM} is easiest to analyze through a non-adaptive statistical property.
Roughly, after a short hashed prefix separates the adversary's queries, the construction becomes perfectly random unless the hash key falls into a bad set.
The notion of almost-randomness captures this situation: the bad set may depend on the queried inputs and on part of the key, but outside the bad set the remaining randomness makes the answers perfectly uniform.

The left-monotonicity requirement is a natural technical condition that allows the above non-adaptive condition to imply security even against adaptive adversaries. Intuitively, it requires that, as an adaptive transcript grows, once the queried set becomes bad, adding more queries cannot make it good again.

\begin{definition}[Left-Monotone Sets]
{\normalfont
Let $\Set X$, $\AUX$ be two sets.
Let $\Set X^*$ be the power set of $\Set X$.
A set $\Set S \subseteq \Set X^* \times \AUX$ is \emph{left-monotone} if for every $(\vc x, aux) \in \Set S$, and every $\vc x' \supseteq \vc x$, it holds that $(\vc x', aux) \in \Set S$.
}

\end{definition}

Informally, a family is almost-random if, except for a small bad event depending on the queried inputs, its outputs on any fixed set of distinct inputs are perfectly uniform.
\begin{definition}[$(q, \epsilon)$-Almost-Random]
{\normalfont
\label{def:almostrandom}
Let $q \in \N$.
Let
$$\ff F = \{ F_{r, aux} \colon \Set X \to \Set Y \}_{(r, aux) \in \Set K}$$
be a function family with key space $\Set K = \Set R \times \AUX$.
$\ff F$ is $(q, \epsilon)$-\emph{almost-random}
if there exists a left-monotone set $\BAD\subseteq\Set X^*\times\AUX$ such that, for every $q'\le q$ and every tuple $\vc x=(x_{(1)},\ldots,x_{(q')})\in\Set X^{q'}$ of distinct inputs (identified with its underlying set), it holds that
\begin{enumerate}
    \item $\Pr_{aux \sample \AUX}[(\vc x, aux) \in \BAD] < \epsilon$.
    \item Fix any $aux \in \AUX$ such that $(\vc x, aux) \notin \BAD$.
    Let $Y\coloneqq F_{r,aux}(\vc x)$, where the randomness is over $r\sample\Set R$. Then $Y$ is uniformly distributed over $\Set Y^{q'}$.
\end{enumerate}
}
\end{definition}

A key advantage of almost-randomness is that it amplifies very cleanly under sum of independent copies.
The next lemma quantifies this amplification.
\begin{lemma}[Sum Amplifies Almost-Randomness]
\label{lemma:XORAmplifiesAlmost-Randomness}
Let $n \in \N$.
If $\ff F \colon \Set X \to \Set Y$ is $(q, \epsilon)$-almost-random, then $\ff F^{\oplus n} \colon \Set X \to \Set Y$ (\Cref{def:sumoffunctionfamily}) is $(q, \epsilon^n)$-almost-random.
\end{lemma}
\begin{proof}
See \Cref{appendix:XORAmplifiesAlmost-Randomness}.
\end{proof}

Finally, we recall the theorem showing that almost-randomness already implies full adaptive pseudorandomness (i.e., strong pseudorandomness).
Together with the previous amplification lemma, this will be the bridge from our non-adaptive analysis to the final PRF guarantee.
\begin{lemma}[Almost-Random $\Rightarrow$ Strong Pseudorandom (Lemma 3.2, \cite{TCC:BermanHKN13,JOC:BermanHKN19})]
\label{lemma:almostrandomimplypseudorandom}
If $\ff F$ is $(q,\epsilon)$-almost-random, then $\ff F$ is $(q,\infty,\epsilon)$-strong pseudorandom.
\end{lemma}

\section{Tapering-GGM: A Generic Construction}
\label{sec:taperingGGM}

This section presents our generic weak-to-strong PRF transformation that drives all the main results of the paper.
The construction can be viewed as a \emph{tapering} version of the GGM tree: as the evaluation moves down the tree, the state length shrinks geometrically, and therefore the depth of each subsequent WPRF call also shrinks.
The security proof has two parts.
First, we use the computational security of WPRF for the top constant number of levels, at which point all queries are shattered with high probability.
Second, below that shattering layer, the property of key-uniformity assures uniformity information-theoretically, so no computational security is required from the WPRFs. The latter is formalized next.

\begin{definition}[Key-Uniform WPRFs]
\label{def:hashingweakprf}
\normalfont
A WPRF $\ffe G$ is \emph{key-uniform} if $\ffe G$ is a $1$-wise independent hashing.
Equivalently, for every security parameter $\lambda$ and every input $r$ in the input domain of $\ff G_\lambda$, the random variable $G_k(r)$ is uniform when $k$ is sampled uniformly from the key space.
\end{definition}

The basic idea to obtain key-uniform WPRFs is simple: XORing the output of a one-bit
WPRF with an independent bit stored in the key makes the output
exactly uniform for every fixed input.
We apply this operation to  independent
copies and concatenate their outputs.

\begin{lemma}[Key-uniformization]
\label{lemma:keyuniformization}
Let $d(\cdot)$ be a non-decreasing function of $\lambda$.
Assume there exists a depth-$d$ WPRF. Then there exists a depth-$(d+O(1))$ key-uniform WPRF
\[
 \ffe G = \left\{\ff G_{\lambda} = \{ G_k \colon \bina^{p(\lambda)} \to \bina^{\lceil\sqrt{\lambda}\rceil}, \; k \in \bina^{\lambda} \} \right\}_{\lambda \in \N},
\]
 where $p(\lambda)$ is polynomially bounded.
 \end{lemma}

\begin{proof}
See \Cref{appendix:lemma:keyuniformization}.
\end{proof}

We now describe the tapering tree that is central to our construction.
There are three global parameters.
The branching factor $\Delta$ is polynomial in the security parameter.
The number of layers $\ell$ is a super-constant, so the number of leaves $\Delta^{\ell}$ is super-polynomial.
The state length at each level is
\[
\lambda^{(i)} \coloneqq \lceil\lambda^{2^{-i+1}}\rceil,
\]
so the state length drops from $\lambda^{(i)}$ to $\lambda^{(i+1)}=\lceil\sqrt{\lambda^{(i)}}\rceil$ at level $i$.
There are $m$ independent such trees in order to amplify almost-randomness by XOR at the end.

\begin{figure}[!ht]
  \begin{framed}
    \begin{center}
      \textbf{$\TGGM$}
    \end{center}

    \medskip

    \textbf{On input} $\tilde{x} \in [\Delta]^{\ell}$:
    \begin{enumerate}
          \item $y^{(1)} \gets \key_{G} \in \bina^{\lambda}$.
          \item For $i = 1, \dots, \ell$  sequentially:
          $$
          y^{(i + 1)} \gets \PRG^{(i)}_{\tilde{x}_i}\left(y^{(i)} \right) \in \bina^{\color{red}{\lambda^{(i+1)}}},
          $$
          where $\PRG^{(i)}_{j} \colon \{0,1\}^{\color{red}{\lambda^{(i)}}} \to \bina^{\color{red}{\lambda^{(i+1)}}}$ for $j \in [\Delta]$ is defined as:
        $$\PRG^{(i)}_{j}(\cdot) \coloneqq G_{\lambda^{(i)}}\left(\cdot, \, r^{(i)}_{j} \right)$$
          \item Output $y^{(\ell + 1)} \in \bina^{\lambda^{(\ell+1)}}$.
      \end{enumerate}
        \medskip
    \textbf{Depth:}
    $$
    d_{\TGGM}(\lambda) = O(\log \lambda) + \sum_{i = 1}^{\ell} d_{\ff G}(\lambda^{(i)}),
    $$
    where $d_{\ff G}(\lambda^{(i)})$ is the depth needed to evaluate the WPRF $\ff G_{\lambda^{(i)}}$.
  \end{framed}
  \caption{The core of the tapering tree. At level $i$, the state has length $\lambda^{(i)}$, and the WPRF call shrinks it to length $\lambda^{(i+1)}$.}
  \label{fig:TGGM}
\end{figure}

\begin{construction}[Generic Construction of PRFs]
\label{construct:framework}
{\normalfont
Let $\lambda$ be a sufficiently large integer security parameter; finitely many smaller parameters are defined separately.
We define the following parameters.
\begin{itemize}
    \item $m = \lambda$, the number of independent tapered trees whose outputs are XORed.
    \item $\ell = \lceil\log\log\log\lambda\rceil$, the number of layers in each tapered tree.
    \item $\Delta = 2^{\lceil\log\lambda\rceil}$, the number of children per node in each layer; thus $\lambda\le\Delta<2\lambda$.
\end{itemize}

\medskip

\noindent \textbf{Building blocks.}
\begin{enumerate}
     \item A key-uniform WPRF $$\ffe{G} = \left\{\ff G_{\lambda} = \{G_k \colon \Set XG_{\lambda} \to \Set {YG}_{\lambda}, \; k \in \Set {KG}_{\lambda}\}\right\}_{\lambda \in \N},$$
     where $\Set XG_{\lambda} = \bina^{p(\lambda)}$ for a polynomially bounded function $p$, and
     $$
     \Set YG_{\lambda} = \bina^{\lceil\sqrt{\lambda}\rceil}, \; \Set KG_{\lambda} = \bina^{\lambda}.
     $$
  \item A 2-wise independent hashing  $$\ff H_{\lambda} = \{ H_k \colon \bina^{\lambda} \to [\Delta]^{\ell}, \; k \in \Set {KH}_{\lambda}\}.$$
\end{enumerate}

\noindent  \textbf{Function specification.}
We will define a function family $\ff F_{\lambda}$:
$$
\ff F_{\lambda} = \{F_{\key_F}\colon\bina^{\lambda}\to\bina^{\lambda^{(\ell+1)}},\ \key_F\in\Set K_{\lambda}\}.
$$
The final PRF is defined as
\[
\ffe{F}' = \{\ff F'_{\lambda}\coloneqq \ff F_{\lambda}^{\oplus m} \}_{\lambda \in \N},
\] where the operator $\oplus$ is defined in \Cref{def:sumoffunctionfamily}.  Thus a key for $\ff F'_{\lambda}$ consists of $m$ independent keys for $\ff F_\lambda$, and the output is their XOR.

\medskip

\noindent  \textbf{Tapering parameters.}
Set $\lambda^{(1)}=\lambda$ and, for $i\in[\ell]$, define
\[
\lambda^{(i+1)}\coloneqq\lceil\sqrt{\lambda^{(i)}}\rceil.
\]
Equivalently, $\lambda^{(i)}=\lceil\lambda^{2^{-i+1}}\rceil$.
In particular, the exact output length is $\lambda^{(\ell+1)}=\lceil\lambda^{2^{-\ell}}\rceil=\lambda^{\Theta(1/\log\log\lambda)}$.

\medskip

\noindent  \textbf{Function key for one tree.}
$$
\Set K_{\lambda} = \Set KG_{\lambda}\times
\prod_{i=1}^{\ell}(\Set XG_{\lambda^{(i)}})^{\Delta}\times\Set KH_{\lambda}
$$
In words, a key for $\ff F_{\lambda}$ is $\key_{F} = (\key_{G}, \key_{XG}, \key_{H})$, where
\begin{itemize}
  \item $\key_{G}  \in \Set KG_{\lambda} = \{0, 1\}^{\lambda}$ is the initial state, equivalently the key for the top-level WPRF $\ff G_{\lambda}$.
  \item $\key_{XG} = \{(r^{(i)}_1,\ldots,r^{(i)}_{\Delta}) \in \Set {XG}_{\lambda^{(i)}}^{\Delta}\}_{i\in[\ell]}$ is the input pool for the WPRFs used at the corresponding levels.
  \item $\key_{H} \in \Set KH_{\lambda}$ is the key for the pairwise independent hash family.
\end{itemize}

\noindent \textbf{Function evaluation.}
On input $x\in\{0,1\}^{\lambda}$, the function $F_{\key_F}$ with key $\key_{F}$ outputs $y \in \bina^{\lambda^{(\ell+1)}}$ as follows:
\[
\tilde{x} \gets H_{\key_{H}}(x) \in [\Delta]^{\ell},
\qquad
y \gets \TGGM(\tilde{x})\in\{0,1\}^{\lambda^{(\ell+1)}},
\]
where the macro $\TGGM(\cdot)$ is defined in \Cref{fig:TGGM}.

}
\end{construction}

\begin{remark}
{\normalfont
The input pool for the WPRFs $G_{\lambda^{(i)}}$ consists of $\Delta$ random elements. Note that $\Delta=\Theta(\lambda)$ is super-polynomial in the key size $\lambda^{(i)}(=\lambda^{o(1)})$ for large enough $i$. This component is conceptually analogous to a PRG with super-polynomial stretch, but does not require any super-polynomial hardness assumption (for security proof).
}
\end{remark}

\begin{lemma}[Depth Preserving]
\label{lemma:depthofTGGM}
For any function $\func d(\lambda) = \Omega(\log \lambda)$ such that $d(\lambda)/\log \lambda$ is non-decreasing,
if the WPRF $\ffe G$ can be computed in depth $d(\lambda)$,
then the function family ensemble  $\ffe {F}'$ defined in \Cref{construct:framework} can also be computed in depth $O(d(\lambda))$.
In particular, if $\ff G$ is in $\mathsf{NC}^1$, then $\ff F'$ is in $\mathsf{NC}^1$.
\end{lemma}
\begin{proof}

Choose $\ff H_\lambda$ by evaluating an affine polynomial over $\mathbb F_{2^\lambda}$ and taking $\ell\log\Delta$ output bits, grouped into $\ell$ blocks of $\log\Delta$ bits.
For sufficiently large $\lambda$, $\ell\log\Delta\le\lambda$; each block encodes an element of $[\Delta]$.
This gives exact pairwise independence, a uniform binary key, and evaluation depth $O(\log\lambda)$.
Once $\tilde{x}=H_{\key_H}(x)$ is known, all elements $r^{(i)}_{\tilde{x}_i}$ can be selected in parallel, with depth $O(\log \Delta)=O(\log\lambda)$.
The only inherently sequential part is the tapering path: level $i$ invokes $G_{\lambda^{(i)}}(\cdot,r^{(i)}_{\tilde{x}_i})$ after the previous level has produced a state.
Thus one tree has depth
\[
O(\log\lambda)+\sum_{i=1}^{\ell} d(\lambda^{(i)}).
\]
The final $m$-fold XOR in $\ff F'_{\lambda}=\ff F_\lambda^{\oplus m}$ adds depth $O(\log m)=O(\log\lambda)$.
Therefore,
\begin{align*}
    \func d_{\ffe {F}'}(\lambda)
    &\le O(\log \lambda) + \sum_{i = 1}^{\ell} d(\lambda^{(i)}) \\
    &\le O(\log \lambda) + \sum_{i = 1}^{\ell} \frac{\log \lambda^{(i)}}{\log \lambda}\cdot d(\lambda) \\
    &\le O(\log\lambda)+\sum_{i=1}^{\ell}\left(2^{-i+1}+\frac{1}{\log\lambda}\right)d(\lambda) \\
    &\le O(d(\lambda)).
\end{align*}
The second line uses the non-decreasing monotonicity of $d(\lambda)/\log\lambda$, and the last line uses $d(\lambda)=\Omega(\log\lambda)$.
\end{proof}

\begin{theorem}[Main theorem]
\label{theorem:main}
Let $m,\Delta,\ell$, and $\lambda^{(i)}$ be as in \Cref{construct:framework}.
For $q \in \N$,  let
\[
c_q \coloneqq \lceil 2\log_{\Delta} q\rceil,
\]
let $\lambda_q$ be the smallest $\lambda$ for which \Cref{construct:framework} is defined and $c_q<\ell$.
There exists a constant $c_{s} > 0$ determined by the size of the key-uniform WPRF $\ffe{G}$ such that, for any $q,t\in\N$ with $t>q^2\lambda^{c_s}$ and any $\lambda\ge\lambda_q$, the following holds:
If for every $i\in[c_q]$, the family $G_{\lambda^{(i)}}$ is $(\Delta,t,\epsilon_i)$-weak pseudorandom, then $\ff F'_{\lambda}$ is $(q,t-q^2\cdot \lambda^{c_s},\epsilon')$-strong pseudorandom, where
\[
\epsilon' = 2^{-m}+\sum_{i=1}^{c_q} q\cdot m\cdot \epsilon_i.
\]
\end{theorem}

\begin{corollary}
\label{corollary:main}
If $\ffe{G}$ is a key-uniform WPRF, then $\ffe{F}'$ is a PRF.
\end{corollary}

\begin{proof}
The proof is a quantifier unpacking of \Cref{theorem:main}: for every polynomial query number and circuit size, the transition level $c_q$ is a constant determined by the polynomial degree, so all WPRF calls used in the hybrid proof have polynomially related security parameters, giving us negligible $\epsilon_i$ for all $i \in [c_q]$.
See \Cref{Appendix:polysecure} for the full argument.

\end{proof}

\begin{proof}[Proof of \Cref{theorem:main}]
High-level idea: The proof follows a two-stage hybrid argument.
The first stage replaces the first $c_q$ levels of each of the $m$ tapered trees by truly random functions, one level and one tree at a time; this is the only place where WPRF (computational) security is used.
The purpose of this stage is to shatter the $q$ queries, so that with high probability they all reach distinct nodes at level $c_q$ (recall that initially they all start from the same node, namely the root of the tree).
In the second stage, we show that, after replacing the first $c_q$ layers, the resulting PRF is statistically close to a random function: conditioned on the bad event not occurring (i.e., all $q$ queries are shattered), the key-uniformity property is used to argue the uniformity of the $q$ leaves (corresponding to the shattered $q$ queries), and hence, the PRF output. Finally, $m$ independent XOR copies amplify the probability of the bad event to negligible  $2^{-m}$. A formal proof follows.

Fix $q, t \in \N$ and fix any $\lambda$ such that $c_q < \ell$.
We define the following hybrids.

\begin{itemize}
      \item Hybrid $\game H_0$: the adversary is given  oracle access to a random element of the function family $\ff F'_{\lambda}$, and can make at most $q$ oracle queries.
      \item Hybrid $\game H_{1, m', j}$, for $m' \in [m], j \in \{0, \dots, c_q\}$: the adversary is given oracle access to a random element of a hybrid function family $\ff F'_{(m', j)}$, which is defined below in \Cref{def:hybridprf}, and can make at most $q$ oracle queries.
      \item Hybrid $\game H_2$: the adversary is given oracle access to a random function $R(\cdot)$, and can  make at most $q$ oracle queries.
\end{itemize}

\begin{definition}[Hybrid function family $\ff F'_{(m', j)}$]
\label{def:hybridprf}
For $m' \in [m], j \in \{0, \dots, c_q\}$, define the function family $$
\ff F'_{(m', j)} \coloneqq \ff F_{(c_q)}^{\oplus (m' - 1)} \oplus  \ff F_{(j)} \oplus
\ff F_{(0)}^{\oplus (m - m')},
$$
where $\ff F_{(j)}$ is defined the same as $\ff F_{\lambda}$ except that instead of invoking $\TGGM$, it invokes $\TGGM_{(j)}$ (\Cref{fig:hybridTGGM}).
\end{definition}

By definition, $\game H_0$ is identical to $\game H_{1,1,0}$, and for every $m'<m$, $\game H_{1,m',c_q}$ is identical to $\game H_{1,m'+1,0}$.

\begin{figure}[!ht]
  \begin{framed}
    \begin{center}
      \textbf{Hybrid $\TGGM_{(j)}$}
    \end{center}

    \medskip

    \textbf{On input} $\tilde{x} \in [\Delta]^{\ell}$:
    \begin{enumerate}
          \item If $j=0$, set $y^{(1)}\gets \key_G$. If $j\ge 1$, set
          $
          y^{(j+1)} \gets R(\tilde{x}_{1\dots j}) \in \bina^{\lambda^{(j+1)}}$,
          where $R$ is a random function treated as part of the key.
          \item For $i=j+1,\dots,\ell$ sequentially, compute
          \[
          y^{(i+1)} \gets \PRG^{(i)}_{\tilde{x}_i}\left(y^{(i)}\right) \in \bina^{\lambda^{(i+1)}},
          \]
          where $\PRG^{(i)}_{h}\colon\bina^{\lambda^{(i)}}\to\bina^{\lambda^{(i+1)}}$ for $h \in [\Delta]$ is defined by
          \[
          \PRG^{(i)}_{h}(\cdot)\coloneqq G_{\lambda^{(i)}}(\cdot,r^{(i)}_h).
          \]
          \item Output $y^{(\ell+1)}\in\bina^{\lambda^{(\ell+1)}}$.
      \end{enumerate}
  \end{framed}
  \caption{Hybrid $\TGGM_{(j)}$, where the first $j$ levels have been collapsed into a random function on length-$j$ prefixes.}
  \label{fig:hybridTGGM}
\end{figure}

\begin{lemma}
\label{lemma:H1indis}
For every $m'\in[m]$ and every $j\in\{0,\dots,c_q-1\}$, the games $\game H_{1,m',j}$ and $\game H_{1,m',j+1}$ are $(q,t-q^2\cdot\lambda^{c_s},q\epsilon_{j+1})$-indistinguishable.
\end{lemma}
\begin{proof}
Suppose that a circuit $D$ of size at most $t-q^2\lambda^{c_s}$ distinguishes $\game H_{1,m',j}$ from $\game H_{1,m',j+1}$ with advantage at least $q\epsilon_{j+1}$.
We build a circuit $D'$ that distinguishes the following games with the same advantage, using $q^2\poly(\lambda)$ additional gates.
Applying the product lemma (\Cref{lemma:productpreserve}) with $q$ independent keys and $\Delta$ samples adds at most $\rho(\Delta,q,s,a,b)=q\poly(\lambda)$ gates, where $s,a,b$ are the size and encoding lengths of $G_{\lambda^{(j+1)}}$.
Choose the constant $c_s$ sufficiently large to cover both costs (uniformly over all levels).
The resulting single-key distinguisher then has size at most $t$ and advantage at least $\epsilon_{j+1}$, contradicting the assumed weak pseudorandomness.
\begin{itemize}
    \item Game $\game G_0$: Sample $\Delta$ random inputs $\vc a = (a_1,\ldots,a_{\Delta}) \sample \Set XG^{\Delta}_{\lambda^{(j+1)}}$, sample $q$ random keys $\vc k = (k_1, \dots, k_q) \sample \Set KG^q$, and give the adversary $G(\vc k, \vc a)$\footnote{In fact, we prove a stronger statement in which \(\vc a\) is not revealed to the adversary. Thus, this challenge may also be viewed as a synthesizer challenge (see \Cref{corollary:Depth-PreservingBootstrappingsynthesizers}).}, where $G(\vc k, \vc{a})\coloneqq (G(k_1, \vc a), \dots, G(k_q, \vc a)) \in \Set {YG}^{q\Delta}_{\lambda^{(j+1)}}$.
    \item Game $\game G_1$: Sample $\Delta$ random inputs $\vc{a} \sample \Set XG^{\Delta}_{\lambda^{(j+1)}}$, sample $q\Delta$ random outputs $\vc u \sample \Set {YG}^{q\Delta}_{\lambda^{(j+1)}}$, and give the adversary $\vc u$.
\end{itemize}

The distinguisher $D'$ obtains $\vc b = (b_{i, j})_{i\in [q], j \in [\Delta]}$ as input, then samples a random key $\vc k$ for $\ff F'_{(m', j + 1)}$ and modifies its $m'$-th fold: It answers $D$'s (hashed) query $\tilde{x}$ by evaluating $\ff F'_{(m', j + 1)}$ with key $\vc k$ and returns the output, except that when evaluating the $m'$-th fold $\ff F_{(j + 1)}$, instead of computing $y^{(j+2)} \gets R(\tilde{x}_{1\dots j+1})$,
it sets $y^{(j+2)} \gets b_{i, \tilde{x}_{j+1}}$,
where the index $i = i(\tilde{x}_{1\dots j}) \in [q]$ is chosen to ensure consistency among $D$'s queries. For example it can be set to the smallest previously unused index when $\tilde{x}_{1\dots j}$ occurs for the first time. This tracking of queries can be implemented with $q^2\cdot \poly(\lambda)$ additional gates.
\begin{itemize}
    \item If $D'$ receives $\vc b$ from Game $\game G_0$, then $y^{(j+2)}=G(k_i,a_{\tilde{x}_{j+1}})$, where $i$ is chosen consistently with $\tilde{x}_{1\dots j}$. Therefore, $D'$ perfectly simulates $H_{1, m', j}$.
    \item If $D'$ receives $\vc b$ from Game $\game G_1$, then $y^{(j+2)}=u_{i,\tilde{x}_{j+1}}$, where $i$ is chosen consistently with $\tilde{x}_{1\dots j}$. Therefore $D'$ perfectly simulates $H_{1, m', j + 1}$.
\end{itemize}
Therefore, $D'$ has the same distinguishing advantage as $D$.
\end{proof}

\begin{lemma}
\label{lemma:H1H2indis}
Games $\game H_{1, m, c_q}$ and $\game H_{2}$ are $(q, \infty, 2^{-m})$-{indistinguishable}.
\end{lemma}
\begin{proof}
Since $\ff F'_{(m,c_q)}=\ff F_{(c_q)}^{\oplus m}$, it suffices by \Cref{lemma:XORAmplifiesAlmost-Randomness} and \Cref{lemma:almostrandomimplypseudorandom} to prove that $\ff F_{(c_q)}$ is $(q,1/2)$-almost-random (\Cref{def:almostrandom}). Intuitively, the only bad event is that two distinct oracle queries hash to the same length-$c_q$ prefix; outside this event, the random function inside $\TGGM_{(c_q)}$ gives independent uniform states, and key-uniformity keeps them uniform down to the leaves.

Let $\AUX=\Set {KH}_{\lambda}$ be the hash key space, and let $\Set R$ denote all remaining key material of $\ff F_{(c_q)}$, including the random function used by $\TGGM_{(c_q)}$ and the input pools. Define
$$
\BAD\coloneqq\{(S,\key_H)\in\Set X^*\times\AUX:\ \exists x\ne x'\in S,\ H_{\key_H}(x)_{1\dots c_q}=H_{\key_H}(x')_{1\dots c_q}\}.
$$
$\BAD$ is left-monotone.

Fix any $q'\le q$ and any tuple of distinct inputs $\vc x=(x_{(1)},\ldots,x_{(q')})\in\Set X^{q'}$.
We verify the two conditions in \Cref{def:almostrandom}.
\begin{enumerate}
    \item By \Cref{lemma:RangeRestrictionpreservepseudorandom} and \Cref{def:kwise}, restricting output range of $\ff H$ from $[\Delta]^{\ell}$ to $[\Delta]^{c_q}$ preserves $2$-wise independent hashing. By \Cref{lemma:collision}, it holds that
    \begin{equation*}
    \begin{split}
        \Pr_{aux\sample\AUX}[(\vc x,aux)\in\BAD]
        &\le \frac{\binom{q'}{2}}{\Delta^{c_q}}
        \le \frac{\binom q2}{q^2}
        =\frac{q-1}{2q}<1/2.
    \end{split}
    \end{equation*}
    \item
    Second, fix any $\key_H$ such that $(\vc x,\key_H)\notin\BAD$ and write $\tilde{x}_{(i)}\coloneqq H_{\key_H}(x_{(i)})$ for $i \in [q']$.
Then the prefixes $(\tilde{x}_{(1)})_{1\dots c_q},\ldots,(\tilde{x}_{(q')})_{1\dots c_q}$ are distinct. Therefore, the corresponding states at level $c_q$
\[
R((\tilde{x}_{(1)})_{1\dots c_q}),\ldots,R((\tilde{x}_{(q')})_{1\dots c_q})
\]
are independent and uniform. For $c_q=0$, interpret $R$ on the empty prefix as the uniform root key.
Condition on all remaining input pools; the following argument holds for every such choice.

For the $i$-th query, define transformation (fixing the WPRF input) acting on the state at level $j=c_q+1, \dots, \ell$ as
\[
T_j^{(i)}(\cdot) \coloneqq G_{\lambda^{(j)}}(\cdot,r^{(j)}_{(\tilde{x}_{(i)})_j})
\]
By key-uniformity of $G$, each $T_j^{(i)}$ is unbiased which maps a uniform input to a uniform output.
Hence
\[
\left(
T_\ell^{(i)}\circ\cdots\circ T_{c_q+1}^{(i)}
\left(R((\tilde{x}_{(i)})_{1\dots c_q})\right)
\right)_{i\in[q']}
\sim U\!\bigl((\bina^{\lambda^{(\ell+1)}})^{q'}\bigr).
\]
This is exactly the joint output distribution of $\ff F_{(c_q)}$ on the fixed tuple $\vc x$.
\end{enumerate}
Thus $\ff F_{(c_q)}$ is $(q,1/2)$-almost-random.
The $m$-fold XOR is therefore $(q,2^{-m})$-almost-random, and \Cref{lemma:almostrandomimplypseudorandom} gives the lemma.
\end{proof}

We now combine the hybrids.
There are $mc_q$ adjacent transitions of the form handled by \Cref{lemma:H1indis}, and the final transition from $\game H_{1,m,c_q}$ to $\game H_2$ is handled by \Cref{lemma:H1H2indis}.
Applying the triangle inequality (\Cref{lemma:triangle}) gives total distinguishing advantage at most
\[
2^{-m}+\sum_{i=1}^{c_q} q\cdot m\cdot\epsilon_i,
\]
against adversaries of size at most $t-q^2\lambda^{c_s}$ making at most $q$ oracle queries.
This proves \Cref{theorem:main}.
\end{proof}

\begin{corollary}[Main result]
\label{corollary:Depth-PreservingBootstrapping}
For any function $\func d(\lambda) = \Omega(\log \lambda)$ such that $d(\lambda)/\log \lambda$ is non-decreasing, if there exist WPRFs computable in depth $\func d(\lambda)$, then there exist PRFs computable in depth $O(\func d(\lambda))$.
In particular, if there exist WPRFs in $\mathsf{NC}^1$, then there exist PRFs in $\mathsf{NC}^1$.
\end{corollary}
\begin{proof}
By \Cref{lemma:keyuniformization}, any depth-$\func d(\lambda)$ WPRF can be converted into a key-uniform WPRF with depth $\func d(\lambda)+O(1)$.
Applying \Cref{construct:framework} to this key-uniform WPRF and using \Cref{lemma:depthofTGGM} gives evaluation depth $O(\func d(\lambda))$.
Security follows from \Cref{corollary:main}.
\end{proof}

\section[\texorpdfstring{$\mathsf{NC}^1$ PRFs from $\LWE$}{NC1 PRFs from LWE}]{$\mathsf{NC}^1$ PRFs from $\LWE$}\label{sec:LWEPRF}

The goal of this section is to instantiate the depth-preserving  transformation of \Cref{corollary:Depth-PreservingBootstrapping} with LWE-based WPRFs whose evaluation is in $\mathsf{NC}^1$.
The section has two logically separate parts.
First, we build a deterministic low-depth circuit that samples, up to negligible statistical error, the rounded Gaussian errors used in Regev's formulation of $\LWE$~\cite{STOC:Regev05,JACM:Regev09}.
Second, we plug this sampler into a chaining construction in the style of Kim~\cite{EC:Kim20}, obtaining WPRFs from polynomial-modulus $\LWE$.
The depth bottleneck is the sampler: once the errors can be generated in $\poly(\log\log\lambda)$ depth, the entire construction fits inside logarithmic depth.

\paragraph{Notation switch.}
Throughout this section, the symbol $q$ denotes the modulus in $\LWE$ (and no longer the number of oracle queries).

We recall the Learning with Errors assumption.

\begin{definition}[Rounded Gaussian Distribution~\cite{STOC:Regev05,JACM:Regev09}]
\normalfont
Let $q \in \N$ and $\alpha \in (0,1)$.
The rounded Gaussian distribution $\chi=\chi_{\alpha q}$ over $\Z_q$ is obtained by sampling $x\sample\mathcal N(0,\alpha^2)$ over $\mathbb R$, outputting $\lfloor qx\rceil$, and reducing the result modulo~$q$.
Equivalently, $\chi$ is the distribution of $\lfloor X\rceil \bmod q$ for $X\sample\mathcal N(0,(\alpha q)^2)$.
\end{definition}

\begin{remark}
\normalfont
We remark that the original formulation of LWE in~\cite{STOC:Regev05} uses continuous Gaussian errors (denoted by $\LWE(\Phi)$). The worst-case-to-average-case hardness reductions of~\cite{STOC:Regev05,STOC:BrakerskiLPRS13} are natively stated for $\LWE(\Phi)$, which trivially reduces to LWE with rounded Gaussian errors (denoted by $\LWE(\overline{\Phi})$) by rounding~\cite[Lemma 4.3]{STOC:Regev05}.
Although in current usage, LWE is more commonly instantiated with discrete Gaussian errors (denoted by $\LWE(D)$), which reduces to $\LWE(\Phi)$ by adding an additional continuous Gaussian~\cite[Claim 3.9]{STOC:Regev05}.
It is also known that $\LWE(\Phi)$ reduces to $\LWE(D)$~\cite[Theorem 3.1]{CRYPTO:Peikert10}.
To conclude, we have
$$
\LWE(D) = \LWE(\Phi) \leq \LWE(\overline{\Phi}).
$$
That is, LWE with rounded Gaussian is as hard as LWE with discrete Gaussian.
\end{remark}

\begin{definition}[$\LWE$ Assumption~\cite{STOC:Regev05,JACM:Regev09}]
\hypertarget{LWE}
{\normalfont
\label{def:lwe}
Let $n, m, q, t\in \N$. Let $\alpha \in (0, 1)$.
Let $\chi = \chi_{\alpha q}$ be a rounded Gaussian distribution.
The $(n, m, q, \alpha, t, \epsilon)$-$\LWE$ assumption states that $\game G_0 \approx_{0,t, \epsilon} \game G_1$, where
\begin{itemize}
    \item \textbf{Game $G_{0}$}: Sample $\mat A \sample \mathbb{Z}_q^{n\times m}$, $\vc s\sample \mathbb{Z}_q^{n}$, and $\vc e\sample \chi^{m}$. Give the adversary $(\mat A, \vc s\cdot \mat A + \vc e)$.
    \item \textbf{Game $G_{1}$}: Sample $\mat A \sample \mathbb{Z}_q^{n\times m}$, and $\vc u \sample \mathbb{Z}_q^{m}$. Give the adversary $(\mat A, \vc u)$.
\end{itemize}
Now let $\func \alpha(\cdot)$ be a function of $n$.
The $\LWE_{\func \alpha}$ assumption states that there is a polynomially bounded power-of-two modulus $\func q(n)$ with $\func\alpha(n)\func q(n)\ge 2\sqrt n$ such that, for every polynomial sample length $\func m(n)$ and every polynomial time bound $\func t(n)$, there exists a negligible function $\funceps(n)$ for which the $(n,\func m(n),\func q(n),\func \alpha(n),\func t(n),\funceps(n))$-$\LWE$ assumption holds for all sufficiently large $n$.
}
\end{definition}

\begin{lemma}[Coupling bound]
\label{lem:couplingbound}
\normalfont
Let $P$ and $Q$ be distributions over a finite set $\Omega$.
Let $(X,Y)$ be any joint distribution over $\Omega\times \Omega$ such that $X\sim P$ and $Y\sim Q$.
\[
\SD(P,Q) \;\le\; \Pr[X\neq Y].
\]
\end{lemma}

\label{sec:PRFfromLWE}
\subsection{Rounded Gaussian Sampler}
\label{sec:roundedgaussiansampler}

We need a deterministic Boolean circuit $S_\chi$ that maps uniform bits to a distribution statistically close to $\chi_{\alpha q}$.
We will rely on the classical Box--Muller transform: two uniform real numbers are mapped to a standard Gaussian using $\ln$, $\sqrt{\cdot}$, and $\cos$.
The technical work in this subsection is to show that each of these operations can be replaced by a finite-precision, low-depth approximation without changing the rounded output except with negligible probability.
The proof is organized as a pipeline: approximate $\ln$ by a polynomial, approximate $\cos$ by a polynomial, approximate $\sqrt{\cdot}$ by Newton iteration, combine them in Box--Muller transform, and finally argue that rounding is stable away from integer half-boundaries.

\paragraph{Notation.} Let $\kappa\in\N$.
Let $\mathbb{R}_{(\kappa)}$ denote the set
$\{t\cdot 2^{-\kappa}, t\in \mathbb{Z}\}$.
For a subset $S \subset \mathbb{R}$, let $S_{(\kappa)}$ denote the set $S\cap \mathbb{R}_{(\kappa)}$.
\begin{itemize}
    \item For $r\in\mathbb{R}$,
    let $\Trunc_{\kappa}(r) \in \mathbb{R}_{(\kappa)}$ be the value obtained by truncating to $\kappa$ fractional bits.
    \item For a bit string $u\in\{0,1\}^{\kappa}$, we (abusively) use $u$ to denote
$\mathsf{int}(u)/2^{\kappa}\in [0, 1)_{(\kappa)}$, where $\mathsf{int}(u) = \sum_{i \in [\kappa]} 2^{i-1} \cdot u_i$.
\end{itemize}

The Box--Muller transform is convenient here because it reduces Gaussian sampling to a small number of elementary functions.
Each elementary function has a standard low-depth approximation, and the singularity of $\ln r_1$ near zero can be isolated into a negligible bad event.

\begin{lemma}[Box--Muller Transform~\cite{box1958note}]
\normalfont
On input $(r_1,r_2)\in(0,1)^2$, define $\Phi:(0,1)^2\to\mathbb{R}$ by
\[
\Phi(r_1,r_2)\coloneqq \sqrt{-2\ln r_1}\cdot \cos(2\pi r_2).
\]
Then $\Phi(r_1, r_2) \sim \mathcal{N}(0, 1)$ for $(r_1, r_2) \sample (0, 1)^2$.
\end{lemma}

For an arbitrary $u_1\in(0,1)$, we first write $u_1=2^{-s}x$ with $x\in[1/2,1)$.
The identity $\ln u_1=-s\ln 2+\ln x$ then reduces the task to approximating $\ln x$ for $x \in [1/2, 1)$.

\begin{lemma}[Polynomial Approximation of $\ln x$]
\label{lemma:lnapprox}
Let $\kappa \in \N$.
There exists a polynomial $p = p_{\ln, \kappa}$ of degree $\kappa$ such that
\[
    \sup_{x\in[1/2,1]}\bigl|\ln x - p(x)\bigr| \le 2^{-\kappa}.
\]
\end{lemma}
\begin{proof}
See \Cref{appendix:lemma:lnapprox}.
\end{proof}

For the cosine term, the input always lies in the fixed interval $[0,2\pi]$, so a truncated Taylor series gives exponentially small error with degree linear in the precision parameter.

\begin{lemma}[Polynomial Approximation of $\cos x$]
\label{lemma:cosapprox}
Let $\kappa \in \N$.
There exists a polynomial $p = p_{\cos, \kappa}$ of degree at most $\kappa$ such that
\[
    \sup_{x\in[0, 2\pi]}\bigl|\cos x - p(x)\bigr| \le 2^{-\Omega(\kappa)}.
\]
\end{lemma}
\begin{proof}
See \Cref{append:lemma:cosapprox}.
\end{proof}

It remains to compute the square root appearing in Box--Muller transform.
We normalize the input to $[1,4)$ and use Newton iteration from the fixed starting point $2$.
The next lemma shows both fast convergence over the reals and stability under truncation.

\begin{definition}[Newton iteration for $\sqrt{\cdot}$]
\label{def:newton}
\normalfont
Let $\kappa\in\N$.
Let $x\in[1,4)$.
Define the iteration map
\[
F_x(t)\;\coloneqq\;\frac12\left(t+\frac{x}{t}\right).
\]
Define a sequence $(t_i)_{i\ge 0}$ by setting $t_0\coloneqq 2$ and
\[
t_{i+1}\;\coloneqq\;F_x(t_i)\qquad\text{for all }i\ge 0.
\]
Or equivalently, $t_{i} \coloneqq  F_x^{(i)}(2)$ for $i\geq 1$.
\end{definition}

\begin{lemma}[Convergence and finite-precision stability]
\label{lemma:newton}
\normalfont
Let all parameters be as in~\Cref{def:newton}.
Let $T = \log \kappa$, $r\coloneqq \sqrt{x}$.
Then:
\begin{enumerate}
    \item $t_i\ge r$ for all $i\ge 0$.
    \item After $T$ iterations, it holds that
    \[
    |t_T-r|\;\le\;2^{-\Omega(\kappa)}.
    \]
    \item Define the finite-precision iteration by
    \[
    \widetilde t_0\coloneqq 2,
    \qquad
    \widetilde t_{i+1}\;\coloneqq\; \widetilde{F}_{x, \kappa}(\widetilde t_i)\qquad\text{for all }i\ge 0,
    \]
    where $\widetilde{F}_{x, \kappa}(\cdot) \coloneqq \Trunc_{\kappa} (F_x(\cdot))$.
    After $T$ iterations, it holds that
    \[
    |\widetilde t_T-r|\;\le\;2^{-\Omega(\kappa)}.
    \]
\end{enumerate}
\end{lemma}

\begin{proof}
See \Cref{append:lemma:newton}.
\end{proof}

\paragraph{Putting everything together.}
The sampler below is the finite-precision Box--Muller pipeline.
Step~2 computes $-2\ln u_1$ after normalizing $u_1$ into $[1/2,1)$; Step~3 computes its square root; Step~4 computes the cosine factor; and Steps~5--6 scale and round the resulting approximate standard Gaussian.
\begin{construction}[Rounded Gaussian Sampler $S_{\chi}$]
\label{construct:rounded-gaussian-sampler}
\normalfont
Let $\kappa\in\N$ and let $\chi=\chi_{\alpha q}$ be the rounded Gaussian distribution over $\mathbb{Z}_q$.
We define the rounded Gaussian sampler
\[
S_{\chi}:\{0,1\}^{\kappa}\to \mathbb{Z}_q
\]
as follows, taking $\kappa$ to be even.
Set the working precision to $P=\kappa^2$. Hardwire dyadic approximations to $\ln 2$, $2\pi$, $\sigma=\alpha q$, and the coefficients of the displayed approximation polynomials with absolute error at most $2^{-P}$. Normalization by powers of two is exact. Except for these normalizations and the explicitly specified Newton truncations, perform arithmetic with $P$ fractional bits, truncating after each operation. Polynomial evaluation uses balanced multiplication and addition trees. Denote the approximation to $\sigma$ by $\widehat\sigma$.
On input $u=(u_1,u_2)$, where $u_1,u_2\in\{0,1\}^{\kappa/2}$ (interpreted as values in $[0, 1)_{(\kappa/2)}$):
\begin{enumerate}
  \item If $u_1 = 0$, output $0$.
  \item Approximate $z_1 = -2\ln u_1$ as follows:
    \begin{enumerate}
        \item Write $u_1=2^{-s}\cdot x$ where $s \in \mathbb Z_{\ge0}$, $x\in[1/2,1)$.
        \item $z_1 \gets \max\{2^{-\kappa},\,-2 \cdot (-s\ln 2 + p_{\ln, \kappa}(x))\}$, using the finite-precision convention above.
    \end{enumerate}
    \item Approximate $z'_1 = \sqrt{z_1}$ as follows:
    \begin{enumerate}
        \item Write $z_1=2^{2s}\cdot x$ where $s\in\Z$, $x\in[1,4)$.
        \item $z'_1 \gets 2^s \cdot \widetilde{F}^{(\log \kappa)}_{x, \kappa}(2)$, where $\widetilde{F}^{(T)}_{x,\kappa}(2)$ denotes $T$ steps of the finite-precision Newton iteration.
    \end{enumerate}
  \item Approximate $z_2 = \cos(2\pi u_2)$ as follows:
  \begin{enumerate}
      \item $z_2 \gets p_{\cos, \kappa}(2\pi u_2)$
  \end{enumerate}
  \item (Box--Muller Transformation) $e \gets z'_1 \cdot z_2 \in \mathbb{R}$.
  \item Compute $w\gets\Trunc_P(\widehat\sigma\cdot e)$ and output $\overline{e}\gets\lfloor w\rceil\bmod q\in\mathbb Z_q$,
  where $\lfloor\cdot\rceil$ denotes rounding to the nearest integer (with a fixed tie-breaking rule).
\end{enumerate}
\end{construction}

For inputs in $\{0,1\}^{\kappa n}$, $S_{\chi}$ is extended by applying it independently to each $\kappa$-bit block, producing an output in $\mathbb{Z}_q^n$.

\begin{lemma}[Sampler Depth]
\label{lemma:sampler-depth}
For $\log q = O(\kappa)$, the function $S_{\chi}$ can be computed by a Boolean circuit of depth $O(\poly(\log\kappa))$.
\end{lemma}

\begin{proof}
The working precision is $P=\kappa^2$. Polynomial powers, coefficients, and intermediate numerators have at most $\poly(\kappa)$ bits; the scaling factor contributes $O(\log q)$ integer bits. Thus all intermediate values have $\poly(\kappa)$ bits when $\log q=O(\kappa)$, and standard arithmetic on them is computable in depth $\poly(\log\kappa)$.
We now account for the depth of each step of \Cref{construct:rounded-gaussian-sampler}.

\begin{enumerate}
    \item For step 1,
    testing whether $u_1=0$ can be done in depth $O(\log\kappa)$.

    \item For step 2(a), writing $u_1=2^{-s}\cdot x$ amounts to computing the position of the most significant $1$. This can be done in depth $O(\poly(\log\kappa))$. For step 2(b), the value $p_{\ln,\kappa}(x)$ is the evaluation of a degree-$\kappa$ polynomial on a $\poly(\kappa)$-bit fixed-point input, which is in depth $O(\poly(\log\kappa))$.

    \item For step 3(a), writing $z_1=2^{2s}\cdot x$ again amounts to finding the position of the most significant $1$.
    This is computable in depth $O(\poly(\log\kappa))$ by the same reasoning as Step 2(a). For step 3(b), $\widetilde{F}^{(\log \kappa)}_{x,\kappa}(2)$ can be computed using $\log \kappa$ iterations,
    and each iteration consists of a constant number of fixed-point additions/multiplications/divisions and a truncation to $\kappa$
    fractional bits. Each such arithmetic operation on $\poly(\kappa)$-bit numbers can be implemented in depth $O(\poly(\log\kappa))$.
    Therefore, Step 3(b) has depth $O(\poly(\log\kappa))$.

    \item For step 4, computing $p_{\cos,\kappa}(2\pi u_2)$ is again polynomial evaluation. As in Step 2(b), this takes depth $O(\poly(\log\kappa))$.

    \item For step 5,
    the multiplication $e=z_1'\cdot z_2$ is a multiplication on $\poly(\kappa)$-bit inputs, hence has depth $O(\poly(\log\kappa))$.

    \item For step 6, multiplying by $\alpha q$ and rounding to the nearest integer can be implemented by fixed-point multiplication followed by extracting the
    integer part and applying the rounding rule; reducing modulo $q$ is standard integer arithmetic.
    Since $\log q = O(\kappa)$, all of these are operations on $\poly(\kappa)$-bit integers and hence have depth $O(\poly(\log\kappa))$.
\end{enumerate}
\end{proof}

The next lemma says that this shallow sampler still has the right distribution: after coupling the finite-precision computation with the ideal Box--Muller transform, the two rounded outputs differ only when either $r_1$ is extremely close to zero or the Gaussian sample lands very close to a rounding boundary.

\begin{lemma}[Rounded Gaussian Sampler Accuracy]
\label{lemma:sampler-sd-real}
Let $\kappa\in\N$ be even and let $\chi=\chi_{\alpha q}$ be the rounded Gaussian distribution over $\mathbb{Z}_q$.
If $\sigma=\alpha q\ge1$ and $\log q=o(\kappa)$, then
\[
\SD\bigl(\chi,\; S_{\chi}(U_{\kappa})\bigr)\;\le\;2^{-\Omega(\kappa)}.
\]
\end{lemma}

\begin{proof}
We bound the statistical distance via an explicit coupling.
The coupling uses the same underlying real randomness for the ideal Gaussian and for the finite-precision sampler, then shows that the two rounded values agree except on negligible events.
Let $(E,E')$ be jointly distributed over $\mathbb{Z}_q\times \mathbb{Z}_q$ as follows.

\noindent\underline{Coupling distribution $(E,E')$:}
\begin{enumerate}
    \item Sample $(r_1,r_2)\sample (0,1)^2$. Let $u_1 \gets \Trunc_{\kappa/2}(r_1)$, $u_2 \gets \Trunc_{\kappa/2}(r_2)$.
    \item Let $e\gets \Phi(r_1,r_2)\in\mathbb{R}$.
    \item Run the sampler on $(u_1,u_2)$, let $w$ be its fixed-point value just before final integer rounding, and set $e'\gets w/\sigma$, where $\sigma=\alpha q$. If $u_1=0$, set $w=e'=0$.
    \item Let $\overline e \gets \lfloor \alpha q\cdot e\rceil \bmod q$,  $\overline e'\gets \lfloor \alpha q\cdot e'\rceil \bmod q$.
    \item Output $(\overline e,\overline e')$.
\end{enumerate}
By definition of the rounded Gaussian distribution, $E\sim \chi$.
Moreover, $E'\sim S_\chi(U_\kappa)$: indeed, if $r\sample (0,1)$ then $\Trunc_{\kappa/2}(r)$ is uniform over $[0, 1)_{(\kappa/2)}$.
Therefore, by the coupling bound (\Cref{lem:couplingbound}),
it suffices to show that
\[
\Pr[E\neq E'] \;\le\; 2^{-\Omega(\kappa)}.
\]
We first control the numerical discrepancy between $e$ and $e'$.

\begin{claim}
\label{clm:boxmuller-stability}
Conditioned on the event $r_1\ge 2^{-\kappa/4}$, it holds that
\[
|e-e'|\le \delta,
\]
for $\delta = 2^{-\Omega(\kappa)}$.
\end{claim}

\begin{proof}
Define
\[
e'' \;\coloneqq\; \Phi(u_1, u_2) \in \mathbb{R}.
\]
By the triangle inequality,
\begin{equation}
\label{eq:triangle}
|e-e'|\;\le\; |e-e''| + |e''-e'|.
\end{equation}
We bound the two terms on the right-hand side.

\noindent \textbf{Bounding $|e'' - e'|$.}
The polynomial approximation errors are $2^{-\Omega(\kappa)}$. The additional errors from the $P$-bit constants and arithmetic are at most $2^{O(\kappa)}\poly(\kappa)2^{-P}=2^{-\Omega(\kappa^2)}$: the polynomial degrees are $O(\kappa)$, their arguments are bounded by a constant, and their power-basis coefficients and intermediate products have magnitude at most $2^{O(\kappa)}$.
For every nonzero input on the $\kappa/2$-bit grid, $-2\ln u_1\ge 2^{1-\kappa/2}$, so the clamp in Step~2 is inactive for sufficiently large $\kappa$.
The inequality $|\sqrt a-\sqrt b|\le\sqrt{|a-b|}$ for $a,b\ge0$ shows that the logarithm approximation changes the square root by at most $2^{-\Omega(\kappa)}$.
By \Cref{lemma:newton}, the normalized Newton error is $2^{-\Omega(\kappa)}$; rescaling multiplies it by at most $\poly(\kappa)$, since $-2\ln u_1=O(\kappa)$.
The cosine approximation has error $2^{-\Omega(\kappa)}$, and both radial factors are $\poly(\kappa)$-bounded. Hence the pre-scaling output differs from $e''$ by $2^{-\Omega(\kappa)}$.
Finally, because $\sigma\ge1$, replacing $\sigma$ by $\widehat\sigma$ and truncating its product changes $w/\sigma$ by at most $\poly(\kappa)2^{-P}$. Thus $|e''-e'|\le2^{-\Omega(\kappa)}$.

\noindent \textbf{Bounding $|e - e''|$.}
Since $r_1\ge2^{-\kappa/4}$ and $|r_1-u_1|\le2^{-\kappa/2}$, every point between $u_1$ and $r_1$ is at least $\Omega(2^{-\kappa/4})$. The mean value theorem for $\ln$ therefore gives
\[
|\ln r_1-\ln u_1|\le O(2^{\kappa/4})2^{-\kappa/2}=O(2^{-\kappa/4}).
\]
Using $|\sqrt a-\sqrt b|\le\sqrt{|a-b|}$, we obtain
\[
\bigl|\sqrt{-2\ln r_1}-\sqrt{-2\ln u_1}\bigr|\le O(2^{-\kappa/8}).
\]
Also, $|\cos(2\pi r_2)-\cos(2\pi u_2)|\le2\pi 2^{-\kappa/2}$ and $\sqrt{-2\ln u_1}=O(\sqrt\kappa)$. Consequently,
\[
|e-e''|\le O(2^{-\kappa/8})+O(\sqrt\kappa)2^{-\kappa/2}=2^{-\Omega(\kappa)}.
\]
Combining this with $|e''-e'|\le 2^{-\Omega(\kappa)}$ in \eqref{eq:triangle}, we obtain that
\[
|e-e'| \le 2^{-\Omega(\kappa)}.
\]
\end{proof}

\medskip
Let $\Delta \coloneqq \alpha q\cdot \delta$. Since $\alpha<1$ and $\log q=o(\kappa)$, we have $\Delta = 2^{-\Omega(\kappa)}$.
Define the boundary set
\[
\BAD \;\coloneqq\; \bigcup_{t\in\mathbb{Z}}\Big[t+\tfrac12-\Delta,\;t+\tfrac12+\Delta\Big].
\]
On the event $r_1 \geq 2^{-\kappa / 4}$ we have $|\alpha q\cdot e - \alpha q\cdot e'|\le \Delta$ by \Cref{clm:boxmuller-stability}.
Therefore, if additionally $\alpha q\cdot e\notin \BAD$, then rounding to the nearest integer is stable under $\Delta$-perturbations, and hence
\[
\lfloor \alpha q\cdot e\rceil = \lfloor \alpha q\cdot e'\rceil,
\]
which implies $E=E'$ (mod~$q$).
Consequently,
\[
\Pr[E\neq E']
\le \Pr[r_1 < 2^{-\kappa /4}] + \Pr[\alpha q\cdot e\in \BAD].
\]
Since $r_1\sample (0,1)$, we have $\Pr[r_1 < 2^{-\kappa /4}] = 2^{-\kappa/4}$.
It remains to bound $\Pr[\alpha q\cdot e\in \BAD]$.

\begin{claim}
\label{clm:gaussian-boundary}
$\Pr[\alpha q\cdot e\in\BAD]\le 2^{-\Omega(\kappa)}$.
\end{claim}

\begin{proof}
Let $\sigma\coloneqq \alpha q$.
Since $e\sim\mathcal{N}(0,1)$, we have $\sigma e\sim \mathcal{N}(0,\sigma^2)$.
Let $\varphi_\sigma$ denote the density of $\mathcal{N}(0,\sigma^2)$.
Then
\[
\Pr[\sigma e\in\BAD]
= \sum_{t\in\mathbb{Z}} \int_{t+\frac12-\Delta}^{t+\frac12+\Delta}\varphi_\sigma(x)\,dx.
\]
Fix $T\coloneqq \lceil \kappa\sigma\rceil$ and split the sum into $|t|\le T$ and $|t|>T$.
For the central part,
\[
\sum_{|t|\le T}\int_{t+\frac12-\Delta}^{t+\frac12+\Delta}\varphi_\sigma(x)\,dx
\le (2T+1)\cdot (2\Delta)\cdot \varphi_\sigma(0)
\le O(\kappa\sigma)\cdot (2\Delta)\cdot \frac{1}{\sigma\sqrt{2\pi}}
\le 2^{-\Omega(\kappa)}.
\]
For the tail part,
\[
\sum_{|t|>T}\int_{t+\frac12-\Delta}^{t+\frac12+\Delta}\varphi_\sigma(x)\,dx
\le \Pr[|X|>T-1] \le \exp(-\Omega(\kappa^2)),
\]
where $X\sim\mathcal{N}(0,\sigma^2)$ and we used a standard Gaussian tail bound together with $T=\Theta(\kappa\sigma)$.
Combining the two bounds and using $\Delta=2^{-\Omega(\kappa)}$ yields
\[
\Pr[\sigma e\in\BAD]\le O(\kappa)\cdot \Delta + \exp(-\Omega(\kappa^2)) \le 2^{-\Omega(\kappa)}.
\]
\end{proof}
Putting everything together,
\[
\Pr[E\neq E'] \le 2^{-\kappa/4} + 2^{-\Omega(\kappa)} \le 2^{-\Omega(\kappa)}.
\]
Therefore, by \Cref{lem:couplingbound},
\[
\SD\bigl(\chi,\; S_{\chi}(U_{\kappa})\bigr)\le \Pr[E\neq E'] \le 2^{-\Omega(\kappa)}.
\]
\end{proof}

For the WPRF construction we use a bounded version of the sampler, which is needed in the WPRFs security proof. Jumping ahead, the highest bit of an LWE sample should be stable after changing the sampled error.
The map below enforces this deterministically by replacing rare large outputs with zero.

\begin{definition}[$B$-Bounded Error Sampler \(S_{\chi}^{(B)}\)]
\label{definition:rounded-gaussian-sampler}
\normalfont
Let \(S_{\chi}\colon \{0,1\}^{\kappa}\to \mathbb{Z}_q\) be the function defined in \Cref{construct:rounded-gaussian-sampler}, and let \(0 < B < q\). Here $|z|$ denotes the minimum absolute value of an integer representative of $z\in\mathbb Z_q$.
We define the \(B\)-bounded version of \(S_{\chi}\), denoted \(S_{\chi}^{(B)}\), by
\[
S_{\chi}^{(B)}(x)\coloneqq
\begin{cases}
S_{\chi}(x), & \text{if } |S_{\chi}(x)| \le B,\\
0, & \text{otherwise}.
\end{cases}
\]
\end{definition}

The low-depth evaluability and statistical accuracy of the sampler continue to hold for the bounded variant.

\begin{corollary}
\label{corollary:bounded-sampler-depth}
If \(\log q = O(\kappa)\), then the function \(S_{\chi}^{(B)}\) can be computed by a Boolean circuit of depth \(O(\poly(\log \kappa))\).
\end{corollary}
\begin{proof}
Compute $S_\chi$, compare the resulting signed representative with $B$, and either keep the value or output $0$.
The comparison and conditional selection have depth $O(\poly(\log\kappa))$ when $\log q=O(\kappa)$, so the claim follows from \Cref{lemma:sampler-depth}.
\end{proof}

\begin{corollary}
\label{corollary:sampler-sd-real}
Let \(\kappa\in\N\) be even, and let \(\chi=\chi_{\alpha q}\) be the rounded Gaussian distribution over \(\mathbb{Z}_q\).
Let \(\sigma=\alpha q\ge1\) and \(B=\sigma\cdot\kappa<q\).
If \(\log q=o(\kappa)\), then
\[
\SD\bigl(\chi,\; S_{\chi}^{(B)}(U_{\kappa})\bigr)\le 2^{-\Omega(\kappa)}.
\]
\end{corollary}

\begin{proof}
Let \(\chi^{(B)}\) denote the \(B\)-bounded rounded Gaussian distribution, obtained by first sampling from \(\chi\), and then replacing the sample by \(0\) whenever its absolute value exceeds \(B\).
By the triangle inequality,
\[
\SD\bigl(\chi,\; S_{\chi}^{(B)}(U_{\kappa})\bigr)
\le
\SD\bigl(\chi,\; \chi^{(B)}\bigr)
+
\SD\bigl(\chi^{(B)},\; S_{\chi}^{(B)}(U_{\kappa})\bigr).
\]

For the first term,
\[
\SD\bigl(\chi,\; \chi^{(B)}\bigr)
\le
\Pr[|\chi|>B].
\]
Moreover, since \(\chi\) is obtained by rounding a Gaussian \(X\sim\mathcal N(0,\sigma^2)\),
\[
\Pr[|\chi|>B]
\le
\Pr[|X|>B-1].
\]
Recalling that \(B=\sigma\kappa\), a standard Gaussian tail bound gives
\[
\Pr[|X|>B-1]
\le
2^{-\Omega(\kappa^2)}.
\]

For the second term, note that both \(\chi^{(B)}\) and \(S_{\chi}^{(B)}(U_{\kappa})\) are obtained from \(\chi\) and \(S_{\chi}(U_{\kappa})\), respectively, by applying the same deterministic post-processing map that sends every value outside \([-B,B]\) to \(0\). Therefore, by monotonicity of statistical distance under post-processing,
\[
\SD\bigl(\chi^{(B)},\; S_{\chi}^{(B)}(U_{\kappa})\bigr)
\le
\SD\bigl(\chi,\; S_{\chi}(U_{\kappa})\bigr)
\le
2^{-\Omega(\kappa)}.
\]

Combining the above bounds proves the corollary.
\end{proof}

\subsection[\texorpdfstring{WPRFs in $\mathsf{NC}^1$ from $\LWE$}{WPRFs in NC1 from LWE}]{WPRFs in $\mathsf{NC}^1$ from $\LWE$}

The $\mathsf{NC}^1$ WPRF construction of Banerjee, Peikert, and Rosen~\cite{EC:BanerjeePR12} requires a super-polynomial LWE modulus.
We now give WPRFs in $\mathsf{NC}^1$ from polynomial-modulus $\LWE$.
The construction follows a chaining paradigm inspired by Kim~\cite{EC:Kim20}. Each layer we compute many (masked) LWE inner products, then use the highest order bits of the previous layer to deterministically sample fresh errors, and then add those errors to the next layer.
We start with $\kappa^{\tau}$ LWE samples, concatenated as a long vector.
The vector length shrinks by a factor of $\kappa$ at every layer, so after $\tau$ layers only one value remains.
The security proof then walks along the chain, replacing one layer at a time by uniform values (under $\LWE$ assumption).

The reason the sampler from the previous subsection is sufficient is that its error is negligible for $\kappa=\Theta(\log^2\lambda)$, while its depth is only $\poly(\log\kappa)=\poly(\log\log\lambda)$.
Thus the $\tau = \omega(1)$ sequential sampler calls still fit within $O(\log\lambda)$ depth.

\begin{construction}[WPRFs from $\LWE$]
\label{construct:weakprffromlwe}
{\normalfont
Let $\lambda \in \N$ be a security parameter.
We define the following parameters.
\begin{itemize}
    \item $\alpha = n^{-\eta}$, for a constant $\eta>0$, the LWE noise-to-modulus ratio.
    \item $n = \lambda^c$ for $c > 1/\eta$, the LWE dimension.
    \item $\kappa = 2\lceil(\log^2\lambda)/2\rceil$, the even bit precision of the rounded Gaussian sampler.
    \item $q=q(n)$ is the polynomially bounded power-of-two modulus from the $\LWE_{\alpha}$ assumption, with $\sigma=\alpha q\ge2\sqrt n$.
    \item $\chi=\chi_\sigma$ is the rounded Gaussian error distribution and $B=\sigma\kappa$ is the error bound. For sufficiently large $\lambda$, $4\alpha\kappa\lambda\le1$ because $c\eta>1$, and hence $q\ge4B\lambda$.

    \item $\tau = \log \log \lambda$, the chaining parameter.
\end{itemize}
The WPRF is defined as
$$
\ff G_{\lambda} = \{G_k \colon \Z_q^{n} \to \Z_q, k \in \Set K_{\lambda}\}.
$$

\medskip

\noindent  \textbf{Shorthand Notations.}
For $z\in\mathbb Z_q$, represented in $\{0,\ldots,q-1\}$, write $\lceil z\rfloor_2=\lfloor2z/q\rfloor\in\{0,1\}$ for its most significant bit, and apply this operation coordinatewise to vectors.
Define a sequence of  parameters
\[
\kappa^{(i)} \;\coloneqq\; \kappa^{\,\tau - i} \qquad\text{for } i \in \{0, \dots, \tau\}.
\]
In particular, $\kappa^{(0)}=\kappa^{\tau}$ and $\kappa^{(\tau)} = 1$.

\medskip

\noindent \textbf{Building Blocks.}
\begin{enumerate}
  \item A ($B$-bounded) rounded Gaussian error sampler $S^{(B)}_{\chi} \colon \bina^{\kappa} \to \Z_q$ (\Cref{definition:rounded-gaussian-sampler}).
\end{enumerate}

When $S^{(B)}_\chi$ is applied to a string $r\in\bina^{\kappa\cdot h}$, we split $r$ into $h$ consecutive $\kappa$-bit blocks and apply $S^{(B)}_\chi$ independently to each block, obtaining an element of $\Z_q^h$.
In particular, since $\kappa^{(i-1)}=\kappa\cdot\kappa^{(i)}$, the expression $S^{(B)}_\chi(r^{(i)})$ in the evaluation algorithm below denotes an element of $\Z_q^{\kappa^{(i)}}$.

\medskip

\noindent  \textbf{Function Key.}
$$
\Set K_{\lambda} =
\prod_{i = 1}^{\tau} \Z_q^{\kappa^{(i)} \times n} \times \Z_q^{\kappa^{(i)}}
$$
In words, a key for $\ff G_{\lambda}$ is $\key_{\ff G} = (\mat S^{(i)}, \vc p^{(i)})_{i \in [\tau]}$, where
\begin{itemize}
  \item $\mat S^{(i)} = (\vc s^{(i)}_1, \dots, \vc s^{(i)}_{\kappa^{(i)}}) \sample \Z_q^{\kappa^{(i)} \times n}$: LWE secrets.
  \item $\vc p^{(i)} \sample \Z_q^{\kappa^{(i)}}$: Padding vector.
\end{itemize}

\medskip

\noindent \textbf{Evaluation.}
On input $\vc a \in \Z_q^{n}$,
\begin{enumerate}
        \item $\vc z^{(i)} \gets \mat S^{(i)} \cdot \vc a + \vc p^{(i)} \in \Z_q^{\kappa^{(i)}}$ for $i \in [\tau]$ in parallel.
        \item $\vc y^{(0)} \gets \vc 0 \in \Z_q^{\kappa^{(0)}}$.
        \item For $i = 1, \dots, \tau$ sequentially:
        \begin{enumerate}
            \item $r^{(i)} \gets \lceil \vc y^{(i-1)}  \rfloor_2 \in \bina^{\kappa^{(i-1)}}$
            \item $\vc e^{(i)} \gets S^{(B)}_{\chi}(r^{(i)}) \in \Z_q^{\kappa^{(i)}}$
            \item $\vc y^{(i)} \gets \vc z^{(i)}  + \vc e^{(i)}  \in \Z_q^{\kappa^{(i)}}$
        \end{enumerate}
        \item Output $y \gets \vc y^{(\tau)} \in \Z_q$.
    \end{enumerate}
}
\end{construction}

The padding vectors $\vc p^{(i)}$ make the terms $\vc z^{(i)}$ uniform.
The only dependence across layers is through the rounded bit string $r^{(i)}$, which selects the deterministic error vector for the next layer.
The bounded-error property ensures that, except with small probability, changing the initial state of the chain does not change the highest bits of $\vc y$; the $\LWE$ hybrids then replace the masked affine terms by uniform values layer by layer.

\begin{lemma}[Log-Depth Evaluation]
\label{lemma:depthofweakprfsfromlwe}
The function family $\ff G_{\lambda}$ defined in \Cref{construct:weakprffromlwe} can be computed in depth
$O(\log n + \tau\cdot (\poly(\log \kappa) + \log \log q)) = O(\log \lambda)$.
\end{lemma}
\begin{proof}
    All affine forms $\vc z^{(i)}=\mat S^{(i)}\vc a+\vc p^{(i)}$ are computed in parallel.
Each such computation is a matrix-vector multiplication over $\Z_q$ and has depth $O(\log n+\log\log q)$ using standard integer arithmetic.
The only sequential part is the chain over $i=1,\ldots,\tau$.
At each layer, extracting $\lceil\vc y^{(i-1)}\rfloor_2$ is constant depth, all coordinatewise calls to $S^{(B)}_\chi$ run in parallel and have depth $\poly(\log\kappa)$ by \Cref{corollary:bounded-sampler-depth}, and the final addition over $\Z_q$ has depth $O(\log\log q)$.
Thus the total depth is
\[
O\bigl(\log n+\tau\cdot(\poly(\log\kappa)+\log\log q)\bigr)=O(\log\lambda),
\]
for the parameters in \Cref{construct:weakprffromlwe}.
\end{proof}

The theorem below formalizes the chaining argument.
There are two losses: the $\tau\cdot\kappa^\tau\cdot\epsilon$ term comes from the $\LWE$ hybrids across all coordinates and layers, while the remaining terms account for the coupling error of the initial chain state and the statistical error of the bounded sampler.

\begin{theorem}
\label{theorem:securityofweakprfsfromlwe}
Let all parameters be as in~\Cref{construct:weakprffromlwe}.
For any $m,t\in\N$, if the $(n,m,q,\alpha,t,\epsilon)$-$\LWE$ assumption holds, then $\ff G_\lambda$ in \Cref{construct:weakprffromlwe} is an $(m,t-m\cdot\poly(\lambda),\epsilon')$-weak pseudorandom function family, where
$$
\epsilon' =  \tau \cdot \kappa^{\tau} \cdot \epsilon + m\cdot \left(\left(\kappa^{\tau}\cdot \frac{1}{\lambda}\right)^{\tau-1} +  \tau \cdot 2^{-\Omega(\kappa)} \right)
$$
The concrete bound applies whenever the displayed circuit-size bound is positive.
In particular, if the $\LWE$ assumption holds, then $\ffe G$ is a secure WPRF.
\end{theorem}

The proof, given in \Cref{proofoftheorem:securityofweakprfsfromlwe}, follows Kim's hybrid strategy~\cite{EC:Kim20} with the bounded sampler above in place of an ideal error sampler.
The sampler accuracy contributes the $2^{-\Omega(\kappa)}$ term, and boundedness is used in the coupling step that compares the real chain with the chain started from a uniform intermediate value.

Combining \Cref{lemma:depthofweakprfsfromlwe} and \Cref{theorem:securityofweakprfsfromlwe}, we have the following corollary.
\begin{corollary}[$\mathsf{NC}^1$ WPRFs from $\LWE$]
\label{corollary:weakprfsfromlwe}
    There exist WPRFs in $\mathsf{NC}^1$ if there exists $\epsilon > 0$, such that $\LWE_{\func \alpha}$ assumption holds for $\func \alpha(n) = 1/n^{\epsilon}$.
\end{corollary}
\begin{proof}
    Under the stated inverse-polynomial $\LWE$ assumption, \Cref{theorem:securityofweakprfsfromlwe} gives weak pseudorandomness for the family in \Cref{construct:weakprffromlwe}; the choice $\kappa=\Theta(\log^2\lambda)$ makes the sampler error negligible.
By \Cref{lemma:depthofweakprfsfromlwe}, the same family is computable in depth $O(\log\lambda)$, and hence lies in $\mathsf{NC}^1$.
\end{proof}

Combining \Cref{corollary:weakprfsfromlwe} and \Cref{corollary:main}, we have the following corollary.
\begin{corollary}[$\mathsf{NC}^1$ PRFs from $\LWE$]
There exist PRFs in $\mathsf{NC}^1$ if there exists $\epsilon > 0$, such that $\LWE_{\func \alpha}$ assumption holds for $\func \alpha(n) = 1/n^{\epsilon}$.
\end{corollary}

\begin{corollary}[$\mathsf{NC}^1$ PRFs from $\GapSVP$]
For every constant $\epsilon>0$, there exist PRFs in $\mathsf{NC}^1$ assuming that $\GapSVP$ is hard for polynomial-size quantum circuits to approximate within a factor of $\tilde O(n^{1+\epsilon})$.
\end{corollary}
\begin{proof}
Set the target noise rate to $\alpha(n)=n^{-\epsilon}$, and let $\beta(n)=\alpha(n)/\log^2 n$. Choose a polynomially bounded power-of-two modulus $q$ with $\beta q\ge2\sqrt n$. Regev's quantum worst-case-to-average-case reduction~\cite{STOC:Regev05}, recalled in~\cite[Theorem 2.16]{STOC:BrakerskiLPRS13}, bases search $\LWE$ at noise rate $\beta$ on quantum hardness of $\GapSVP$ within factor $\tilde O(n/\beta)=\tilde O(n^{1+\epsilon})$. The power-of-two search-to-decision reduction~\cite[Theorem 2.17]{STOC:BrakerskiLPRS13} increases the noise rate by a factor $\log^2 n$, giving decisional $\LWE$ at rate $\alpha$. Rounding the continuous errors gives the rounded-Gaussian formulation used here. The previous corollary then gives the claim.
\end{proof}

\section[\texorpdfstring{$\mathsf{NC}^1$ PRFs from $\LPN$}{NC1 PRFs from LPN}]{$\mathsf{NC}^1$ PRFs from $\LPN$}\label{sec:LPNPRF}

\label{sec:LPNassumption}
This section gives an $\mathsf{NC}^1$ PRF from $\LPN$. To this end, we combine the recent construction of WPRFs in $\mathsf{NC}^1$ from standard $\LPN$ due to Ding, Jain, and Komargodski~\cite{STOC:DingJK25} with our generic depth-preserving weak-to-strong transformation from \Cref{corollary:Depth-PreservingBootstrapping}.

 We start by recalling the Learning Parity with Noise assumption and then state the corollary.
\begin{definition}[$\LPN$ Assumption~\cite{C:BlumFKL93}]
\hypertarget{LPN}
{\normalfont
\label{def:lpn}
Let $n,m,t\in\N$, and let $\mu\in(0,1)$.
Let $\Ber_{\mu}$ denote the Bernoulli distribution with parameter $\mu$, namely $\Pr[\Ber_{\mu}=1]=\mu$ and $\Pr[\Ber_{\mu}=0]=1-\mu$.
The $(n,m,\mu,t,\epsilon)$-$\LPN$ assumption states that $\game G_0 \approx_{0,t,\epsilon} \game G_1$, where
\begin{itemize}
    \item Game $\game G_0$: Sample $\mat A\sample \mathbb{F}_2^{n\times m}$, $\vc s\sample \mathbb{F}_2^n$, and $\vc e\sample \Ber_\mu^m$. Give the adversary $(\mat A,\vc s\cdot \mat A+\vc e)$.
    \item Game $\game G_1$: Sample $\mat A\sample \mathbb{F}_2^{n\times m}$ and $\vc u\sample \mathbb{F}_2^m$. Give the adversary $(\mat A,\vc u)$.
\end{itemize}
Now let $\func\mu=\func\mu(n)$ be a function of $n$.
The $\LPN_{\func\mu}$ assumption states that for every polynomial $\func m=\func m(n)$ and every polynomial $\func t=\func t(n)$, there exists a negligible function $\funceps=\funceps(n)$ such that the $(n,m(n),\mu(n),t(n),\epsilon(n))$-$\LPN$ assumption holds for all $n\in\N$.
}
\end{definition}

The following theorem of Ding, Jain, and Komargodski~\cite{STOC:DingJK25}  provides the WPRFs in $\mathsf{NC}^1$ from $\mathsf{LPN}$, as needed.
\begin{theorem}[{\cite[Theorem 1]{STOC:DingJK25}}]
\label{thm:DJK25-lpn-wprf}
If there exists $\epsilon>0$ such that the $\LPN_{\func\mu}$ assumption holds for $\func\mu(n)=n^{-\epsilon}$, then there exist WPRFs in $\mathsf{NC}^1$.
\end{theorem}

\begin{corollary}[$\mathsf{NC}^1$ PRFs from $\LPN$]
There exist PRFs in $\mathsf{NC}^1$ if there exists $\epsilon>0$ such that the $\LPN_{\func\mu}$ assumption holds for $\func\mu(n)=n^{-\epsilon}$.
\end{corollary}
\begin{proof}
By \Cref{thm:DJK25-lpn-wprf}, the stated $\LPN$ assumption implies the existence of WPRFs in $\mathsf{NC}^1$.
Applying our generic depth-preserving weak-to-strong transformation (\Cref{corollary:Depth-PreservingBootstrapping}) upgrades these WPRFs to PRFs in $\mathsf{NC}^1$.
\end{proof}

\section[\texorpdfstring{$\mathsf{NC}^1$ PRFs from $\CDH$}{NC1 PRFs from CDH}]{$\mathsf{NC}^1$ PRFs from $\CDH$}\label{sec:CDHPRF}

\label{sec:CDH}

This section gives an $\mathsf{NC}^1$ PRF from $\CDH$. To this end, we observe that our generic transformation can be instantiated with an object even weaker than a WPRF.
Indeed, the proof of \Cref{theorem:main} never exposes the random input pool of the underlying primitive to the adversary.
As a result, the relevant security notion is not full weak pseudorandomness, but rather pseudorandomness on \emph{hidden} random inputs shared across many independently sampled keys.
This is exactly the notion captured by synthesizers~\cite{FOCS:NaorR95,JCSS:NaorR99}.

We recall the notion of synthesizers, which can be viewed as a relaxation of weak pseudorandom functions.
To motivate the definition, consider the following intermediate notion, which we call \emph{hidden-input weak pseudorandomness}.
Here, the adversary does not receive input-output pairs $\{(x^{(i)},G(k,x^{(i)}))\}$; instead, it only receives outputs of the form $\{G(k,x^{(i)}) \}$ on uniformly random inputs~$x^{(i)}$.
Taken in isolation, this requirement is too weak: for example, the identity map would look perfectly random if the inputs were hidden.
The key additional requirement is therefore that the \emph{same} random input set be reused across many independently sampled keys.
This shared-input condition is exactly what rules out such degenerate examples, and it is also exactly the pattern that appears in the hybrid step of our Tapering-GGM proof.\footnote{This property is automatic in the standard notion of weak pseudorandomness, since the inputs are public and a standard hybrid argument can be applied (\Cref{lemma:productpreserve}).}

\begin{definition}[$(m, q, t, \epsilon)$-Hidden-Input Weak Pseudorandom]
\label{def:IHpseudorandom}
\normalfont

Let $m,q,t \in \N$ and $\epsilon \ge 0$.
A function family $\ff G=\{G_k\colon \Set{X}\to \Set{Y}\}_{k\in\Set{K}}$ is \emph{$(m,q,t,\epsilon)$-hidden-input pseudorandom} if the following two games are $(0,t,\epsilon)$-indistinguishable:
\begin{itemize}
    \item Game $\game G_0$: Sample $q$ random inputs $\vc x=(x^{(1)},\ldots,x^{(q)}) \sample \Set X^q$, sample $m$ random keys $\vc k=(k_1,\ldots,k_m) \sample \Set K^m$, and give the adversary
    \[
    G(\vc k,\vc x) \coloneqq (G(k_1,x^{(1)}),\ldots,G(k_i,x^{(j)}),\ldots,G(k_m,x^{(q)})) \in \Set Y^{mq}.
    \]
    \item Game $\game G_1$: Sample $mq$ random outputs $\vc u \sample \Set Y^{mq}$ and give the adversary $\vc u$.
\end{itemize}
\end{definition}

Synthesizers are defined analogously to WPRFs (\Cref{def:weakprf}), except that the underlying security notion is hidden-input weak pseudorandomness rather than weak pseudorandomness.

\begin{definition}[Synthesizers~\cite{FOCS:NaorR95,JCSS:NaorR99}]
\label{def:synthesizers}
\normalfont
Let $\func m(\cdot)$, $\func q(\cdot)$, $\func t(\cdot)$, and $\funceps(\cdot)$ be functions of the security parameter $\lambda$.
An efficient function family ensemble $\ffe S=\{\ff S_\lambda\}_{\lambda\in\N}$ is a \emph{$(\func m,\func q,\func t,\funceps)$-synthesizer} if for every $\lambda\in\N$, the function family $\ff S_\lambda$ is $(m(\lambda),q(\lambda),t(\lambda),\funceps(\lambda))$-hidden-input pseudorandom.

We say that $\ffe S$ is a (\emph{polynomially secure}) \emph{synthesizer} if for any polynomials $\func m(\cdot)$, $\func q(\cdot)$, and $\func t(\cdot)$, there exists a negligible function $\funceps(\cdot)$ such that $\ffe S$ is a $(\func m,\func q,\func t,\funceps)$-synthesizer.
\end{definition}

Every WPRF is also a synthesizer: one first applies the product lemma (\Cref{lemma:productpreserve}) to reuse the same random input set across many independent keys, and then simply hides those inputs from the adversary.
The converse need not hold, because hiding the inputs can make a function family look pseudorandom even when it fails the standard WPRF definition.
Moreover, by the same one-time-pad transformation used earlier to turn WPRFs into key-uniform WPRFs in \Cref{lemma:keyuniformization}, a synthesizer can be made key-uniform with only a depth-$O(1)$ overhead: one XORs the output with a fresh key-dependent pad, which preserves hidden-input pseudorandomness while enforcing $1$-wise independence.

The relevance to \Cref{construct:framework} is that, throughout the proof of \Cref{theorem:main}, the ``input pool'' to the underlying primitive is sampled as part of the PRF key and is never revealed.
In particular, the only place where security of the underlying family is used is the hybrid step in \Cref{lemma:H1indis}, and there the reduction gives the distinguisher values of the form $(G(k_i,a_j))_{i\in[m],j\in[\Delta]}$ for hidden random inputs $a_1,\ldots,a_\Delta$.
This is precisely a synthesizer challenge.
Consequently, the proof of \Cref{theorem:main} goes through verbatim with a key-uniform synthesizer in place of the underlying key-uniform WPRF.
The key-length and output-length normalization also carries over: independent output copies correspond to additional rows of the same synthesizer challenge, which can be regrouped without revealing its inputs.

\begin{corollary}
\label{corollary:Depth-PreservingBootstrappingsynthesizers}
For any function $\func d(\lambda)=\Omega(\log\lambda)$ such that $d(\lambda)/\log\lambda$ is non-decreasing, if there exist synthesizers computable in depth $\func d(\lambda)$, then there exist PRFs computable in depth $O(\func d(\lambda))$.
In particular, if there exist synthesizers in $\mathsf{NC}^1$, then there exist PRFs in $\mathsf{NC}^1$.
\end{corollary}

We now plug in the classical construction of Naor and Reingold~\cite{FOCS:NaorR95,JCSS:NaorR99}, who showed how to build synthesizers in $\mathsf{NC}^1$ from the Computational Diffie--Hellman assumption.
For completeness, we restate the assumption.

\begin{definition}[$\CDH$ Assumption~\cite{TIT:DiffieH76}]
\label{def:cdh}
\normalfont
Let $\mathbb{G}$ be a (suitably chosen) cyclic group of prime order $p$ with generator $g$, and let $\lambda\coloneqq \log p$ denote the security parameter.
The $\CDH$ assumption states that for every efficient adversary $\att A$, there exists a negligible function $\funceps=\funceps(\lambda)$ such that for every $\lambda\in\N$,
\[
\Pr_{a,b\sample \mathbb{Z}_p}\!\left[\att A(g,g^a,g^b)=g^{ab}\right] \le \epsilon(\lambda).
\]
\end{definition}

\begin{corollary}[$\mathsf{NC}^1$ PRFs from $\CDH$]
There exist PRFs in $\mathsf{NC}^1$ if the $\CDH$ assumption holds.
\end{corollary}

\section*{Acknowledgements}
Ding and Komargodski were supported in part by a grant from the Israel Science Foundation (ISF Grant No. 1774/20), and by the European Union (ERC, SCALE,101162665).
Jain was supported in part by a Google Faculty Research Scholarship, Amazon, Cylab and QCIT of CMU, 0xParc, a Stellar Foundation Research Grant, and the NSF career award.
Views and opinions expressed are
however those of the author(s) only and do not necessarily reflect those of the
European Union or the European Research Council. Neither the European Union
nor the granting authority can be held responsible for them.

\section*{AI Disclosure}
The constructions and main proof ideas were conceived by the authors.
We used ChatGPT 5.5 and 5.6 as an editorial tool to assist with the write-up. In addition, it also produced typesetting of initial versions of some proofs where the proofs themselves were suggested by the authors.
The authors take full responsibility for the correctness, originality, and integrity of the manuscript, including all reported results and findings.

\newpage
\appendix

\section{Omitted Proofs}

\subsection[\texorpdfstring{Proof of \Cref{lemma:XORAmplifiesAlmost-Randomness}}{Proof of almost-randomness amplification}]{Proof of \Cref{lemma:XORAmplifiesAlmost-Randomness}}
\label{appendix:XORAmplifiesAlmost-Randomness}
\begin{lemma}[Restating \Cref{lemma:XORAmplifiesAlmost-Randomness}]
Let $n \in \N$.
If $\ff F$ is $(q, \epsilon)$-almost-random, then $\ff F^{\oplus n}$ is $(q, \epsilon^n)$-almost-random.
\end{lemma}
\begin{proof}
Denote by $\Set R' = \Set R^n$, $\AUX' = \AUX^n$.
$\ff F' = \ff F^{\oplus n}$.
By definition, it follows that
$$
\ff F'= \{F_{r_1, aux_1} \oplus \cdots \oplus F_{r_n, aux_n}: ((r_1, \dots, r_n), (aux_1, \dots, aux_n) )\in \Set R' \times  \AUX'\}
$$
We define
$$
\BAD' \coloneqq \{(\vc x, (aux_1, \dots, aux_n)): \forall i \in [n], (\vc x, aux_i) \in \BAD \}
$$
It follows that $\BAD'$ is left-monotone since $\BAD$ is left-monotone.

Fix any $q'\le q$ and any $\vc x\in\mathcal X^{q'}$ with distinct entries. We prove that conditions 1 and 2 hold for $\ff F'$ and $\BAD'$.
\begin{enumerate}
    \item For any  $aux' \coloneqq (aux_1, \dots, aux_n) \sample \AUX'$, it holds that
    $$
    \Pr[(\vc x, aux') \in \BAD'] = \Pr[ \forall i \in [n], (\vc x, aux_i) \in \BAD] =  \prod_{i \in [n]} \Pr[(\vc x, aux_i) \in \BAD] < \epsilon^n,
    $$
    where the second equality is because $aux_i$ are independent random variables.
    \item Fix any $aux' \coloneqq (aux_1, \dots, aux_n) \in \AUX'$ such that $(\vc x, aux') \notin \BAD'$. It  follows that there exists $i \in [n]$, such that $(\vc x, aux_i) \notin \BAD$.
    Then, for $r' \coloneqq (r_1, \dots, r_n) \sample \Set R'$, $F_{r',aux'}(\vc x)\sim U(\Set Y^{q'})$ because the $i$-th fold $F_{r_i,aux_i}(\vc x)\sim U(\Set Y^{q'})$ and $r_1, \dots, r_n$ are independent.
\end{enumerate}
\end{proof}

\subsection[\texorpdfstring{Proof of \Cref{lemma:keyuniformization}}{Proof of key-uniformization}]{Proof of \Cref{lemma:keyuniformization}}
\label{appendix:lemma:keyuniformization}
\begin{lemma}[Restating \Cref{lemma:keyuniformization}]
Let $d(\cdot)$ be a non-decreasing function of $\lambda$.
Assume there exists a depth-$d$ WPRF. Then there exists a depth-$(d+O(1))$ key-uniform WPRF
\[
 \ffe G = \left\{\ff G_{\lambda} = \{ G_k \colon \bina^{p(\lambda)} \to \bina^{\lceil\sqrt{\lambda}\rceil}, \; k \in \bina^{\lambda} \} \right\}_{\lambda \in \N},
\]
 where $p(\lambda)$ is polynomially bounded.
 \end{lemma}

\begin{proof}
Set $u(\lambda)=\lceil\sqrt{\lambda}\rceil$ and $b(\lambda)=\lfloor\lambda/u(\lambda)\rfloor$.
For all sufficiently large $\lambda$, choose a polynomially related parameter $n(\lambda)\le\lambda$ such that the original WPRF key length at $n(\lambda)$ is at most $b(\lambda)-1$; this is possible because that key length is polynomially bounded.
Restrict the output to one bit, and pad the key with unused bits to obtain a family
\[
\ff F_\lambda=\{F_k\colon\bina^{p(\lambda)}\to\bina,\ k\in\bina^{b(\lambda)-1}\}.
\]
Polynomially related reparameterization preserves WPRF security, and its depth is at most $d(n(\lambda))\le d(\lambda)$.
For a fresh key bit $z$, set $G^{(1)}_{k,z}(x)=F_k(x)\oplus z$.
For every fixed input, this output is uniform over the enlarged key, and the direct-sum lemma (\Cref{lemma:superposition}) preserves WPRF security with polynomial overhead.
Take $u(\lambda)$ independent copies and concatenate their outputs.
The product lemma (\Cref{lemma:productpreserve}) preserves WPRF security with polynomial loss, while each output bit remains independent and uniform over the key.
These copies use $u(\lambda)b(\lambda)\le\lambda$ key bits; pad with unused bits to obtain exactly $\lambda$ key bits.
The copies are evaluated in parallel and XOR has constant depth, giving depth $d(\lambda)+O(1)$.
\end{proof}

\subsection[\texorpdfstring{Proof of \Cref{corollary:main}}{Proof of polynomial security}]{Proof of \Cref{corollary:main}}

\label{Appendix:polysecure}

\begin{corollary}[Restating \Cref{corollary:main}]
If $\ffe{G} = \{\ff G_{\lambda}\}_{\lambda \in \N}$ is a key-uniform WPRF, then $\ffe{F} = \{\ff {F}'_{\lambda}\}_{\lambda \in \N}$ is a PRF.
\end{corollary}

\begin{proof}
Efficiency follows from the construction. To prove security,
fix arbitrary polynomially bounded query and circuit-size bounds
$q(\lambda)$ and $t(\lambda)$.
Choose constants $a,b \ge 1$ such that, for all sufficiently
large $\lambda$,
\[
    q(\lambda) \le \lambda^a
    \qquad\text{and}\qquad
    t(\lambda) \le \lambda^b.
\]
Let $L \coloneqq \lceil 2a\rceil$ and
$\beta \coloneqq 2^{1-L}>0$.
Since $\Delta\ge\lambda$, the transition level satisfies
\[
    c_{q(\lambda)}
    =\left\lceil2\log_\Delta q(\lambda)\right\rceil
    \le\left\lceil2\log_\lambda q(\lambda)\right\rceil\le L.
\]
As $\ell=\lceil\log\log\log\lambda\rceil$ tends to infinity, we have
$c_{q(\lambda)}<\ell$ for all sufficiently large $\lambda$.
Moreover, for every $i\in[L]$,
\[
    \lambda^{(i)}
    =\lceil\lambda^{2^{1-i}}\rceil
    \ge\lambda^\beta.
\]

To account for the reduction overhead in \Cref{theorem:main},
define
\[
    \widehat{t}(\lambda)
    \coloneqq
    \left\lceil
        t(\lambda)+q(\lambda)^2\lambda^{c_s}
    \right\rceil.
\]
Choose a positive integer $C$ such that
\[
    \beta C>\max\{1,b,2a+c_s\}.
\]
Then, for all sufficiently large $\lambda$ and every $i\in[L]$,
\[
    \Delta<2\lambda\le\bigl(\lambda^{(i)}\bigr)^C
    \qquad\text{and}\qquad
    \widehat{t}(\lambda)\le \bigl(\lambda^{(i)}\bigr)^C.
\]
Since $\ffe{G}$ is a WPRF, there exists a negligible function
$\eta$ such that $G_n$ is
$(n^C,n^C,\eta(n))$-weak pseudorandom for every security
parameter $n$.
Consequently, for every $i\in[c_{q(\lambda)}]$, the family
$G_{\lambda^{(i)}}$ is
\[
    \bigl(\Delta,\widehat{t}(\lambda),
          \eta(\lambda^{(i)})\bigr)
    \text{-weak pseudorandom}.
\]

Applying \Cref{theorem:main} with these parameters shows that
$F'_\lambda$ is
$(q(\lambda),t(\lambda),\epsilon'(\lambda))$-strong
pseudorandom for all sufficiently large $\lambda$, where
\[
    \epsilon'(\lambda)
    \coloneqq
    2^{-\lambda}
    +q(\lambda)\lambda
      \sum_{i=1}^{c_{q(\lambda)}}
        \eta\bigl(\lambda^{(i)}\bigr).
\]
It remains to verify that $\epsilon'$ is negligible.
Fix any constant $r>0$.
By negligibility of $\eta$, for all sufficiently large $n$,
\[
    \eta(n)\le n^{-(r+a+2)/\beta}.
\]
Since $\lambda^{(i)}\ge\lambda^\beta$ for every $i\in[L]$,
it follows that, for all sufficiently large $\lambda$,
\[
    \epsilon'(\lambda)
    \le
    2^{-\lambda}
    +L\lambda^{a+1}\lambda^{-(r+a+2)}
    =
    2^{-\lambda}+L\lambda^{-r-1}
    <
    \lambda^{-r}.
\]
Thus $\epsilon'$ is negligible.
Since $q$ and $t$ were arbitrary, $\ffe{F}'$ is a PRF.
\end{proof}

\subsection[\texorpdfstring{Proof of \Cref{lemma:lnapprox}}{Proof of lemma:lnapprox}]{Proof of \Cref{lemma:lnapprox}}

\label{appendix:lemma:lnapprox}

\begin{lemma}[Restating \Cref{lemma:lnapprox}]
Let $\kappa \in \N$.
There exists a polynomial $p = p_{\ln, \kappa}$ of degree $\kappa$ such that
\[
    \sup_{x\in[1/2,1]}\bigl|\ln x - p(x)\bigr| \le 2^{-\kappa}.
\]
\end{lemma}
\begin{proof}
Write $x=1-y$ where $y\in[0,1/2]$. Then by Taylor expansion of $\ln(1-y)$ at $y=0$, we have
\[
\ln x=\ln(1-y)= -\sum_{j=1}^{\infty}\frac{y^{j}}{j}.
\]
Let the degree-$\kappa$ polynomial $p$ be
\[
p(x)\;\coloneqq\; -\sum_{j=1}^{\kappa}\frac{(1-x)^{j}}{j},
\]
The approximation error is
\[
\Bigl|\ln(x)-p(x)\Bigr|
= \sum_{j=\kappa+1}^{\infty}\frac{y^{j}}{j}
\le \sum_{j=\kappa+1}^{\infty} y^{j}
\leq  \sum_{j=\kappa+1}^{\infty} (1/2)^j \leq 2^{-\kappa}
\]
\end{proof}

\subsection[\texorpdfstring{Proof of \Cref{lemma:cosapprox}}{Proof of lemma:cosapprox}]{Proof of \Cref{lemma:cosapprox}}

\label{append:lemma:cosapprox}

\begin{lemma}[Restating \Cref{lemma:cosapprox}]
Let $\kappa \in \N$.
There exists a polynomial $p = p_{\cos, \kappa}$ of degree at most $\kappa$ such that
\[
    \sup_{x\in[0, 2\pi]}\bigl|\cos x - p(x)\bigr| \le 2^{-\Omega(\kappa)}.
\]
\end{lemma}
\begin{proof}
Use the Taylor expansion at $0$:
\[
\cos x=\sum_{j=0}^{\infty}(-1)^j \frac{x^{2j}}{(2j)!}.
\]
Let the degree-at-most-$\kappa$ polynomial $p$ be
\[
p(x)\;\coloneqq\;\sum_{j=0}^{\lfloor\kappa / 2\rfloor}(-1)^j \frac{x^{2j}}{(2j)!},
\]
Denote
\[
a_j(x)\coloneqq \frac{x^{2j}}{(2j)!}.
\]
For $x\in[0,2\pi]$, $j \geq 4$,
\[
\frac{a_{j+1}(x)}{a_j(x)}
=
\frac{x^2}{(2j+2)(2j+1)}
\le
\frac{(2\pi)^2}{(2j+2)(2j+1)}
\le 1.
\]
Therefore, the alternating-series remainder is bounded by the first omitted term:
\[
|\cos x - p(x)|
\le \frac{x^{2\lfloor\kappa/2\rfloor+2}}{(2\lfloor\kappa/2\rfloor+2)!}
\le \frac{(2\pi)^{2\lfloor\kappa/2\rfloor+2}}{(2\lfloor\kappa/2\rfloor+2)!}.
\]
Writing $d=2\lfloor\kappa/2\rfloor+2$ and using Stirling's bound $d!\ge(d/e)^d$, we obtain
\[
\frac{(2\pi)^d}{d!}\le\left(\frac{2\pi e}{d}\right)^d\le2^{-\Omega(\kappa)}.
\]

\end{proof}

\subsection[\texorpdfstring{Proof of \Cref{lemma:newton}}{Proof of lemma:newton}]{Proof of \Cref{lemma:newton}}

\label{append:lemma:newton}

\begin{lemma}[Restating \Cref{lemma:newton}]
\normalfont
Let all parameters be as in~\Cref{def:newton}.
Let $T = \log \kappa$, $r\coloneqq \sqrt{x}$.
Then:
\begin{enumerate}
    \item $t_i\ge r$ for all $i\ge 0$.
    \item After $T$ iterations, it holds that
    \[
    |t_T-r|\;\le\;2^{-\Omega(\kappa)}.
    \]
    \item Define the finite-precision iteration by
    \[
    \widetilde t_0\coloneqq 2,
    \qquad
    \widetilde t_{i+1}\;\coloneqq\; \widetilde{F}_{x, \kappa}(\widetilde t_i)\qquad\text{for all }i\ge 0,
    \]
    where $\widetilde{F}_{x, \kappa}(\cdot) \coloneqq \Trunc_{\kappa} (F_x(\cdot))$.
    After $T$ iterations, it holds that
    \[
    |\widetilde t_T-r|\;\le\;2^{-\Omega(\kappa)}.
    \]
\end{enumerate}
\end{lemma}

\begin{proof}
We prove them one by one.
\begin{enumerate}
    \item Since $t_0=2\ge r>0$, the identity
\[
F_x(t)-r=\frac{(t-r)^2}{2t}\ge0\qquad(t>0)
\]
shows by induction that $t_i\ge r$ for all $i\ge0$.

\item Write $t_i=r(1+\varepsilon_i)$, where $\varepsilon_i\coloneqq (t_i-r)/r\ge 0$.
Then
\[
t_{i+1}
=\frac12\left(r(1+\varepsilon_i)+\frac{r^2}{r(1+\varepsilon_i)}\right)
=\frac r2\left((1+\varepsilon_i)+\frac{1}{1+\varepsilon_i}\right).
\]
Dividing by $r$ and subtracting $1$ gives
\[
\varepsilon_{i+1}
=\frac12\left((1+\varepsilon_i)+\frac{1}{1+\varepsilon_i}\right)-1
=\frac{\varepsilon_i^2}{2(1+\varepsilon_i)}
\le \frac{\varepsilon_i^2}{2},
\]
where the inequality uses $1+\varepsilon_i\ge 1$.
In particular, the relative error decreases quadratically. Since $t_0=2$ and $r\in[1,2)$, we have $\varepsilon_0=(2-r)/r\le 1$, so $\varepsilon_1 \leq 1/2$, and iterating the bound yields $\varepsilon_T\le 2^{-\Omega(2^T)}$. For $T= \log\kappa$, this implies
$|t_T-r|=r\cdot \varepsilon_T\le 2^{-\Omega(\kappa)}$.

\item We prove by induction that $|\widetilde t_i-t_i|\le 2^{-\kappa+1}$ for all $i\ge 0$; the claim then follows by combining this with above.
The base case $i=0$ is immediate since $\widetilde t_0=t_0=2$.

Now assume the claim holds for some $i\ge 0$. We have
\[
F'(t)=\frac12\left(1-\frac{x}{t^2}\right).
\]
Since $t_i\ge\sqrt x$, the induction hypothesis gives $\widetilde t_i\ge\sqrt x-2^{-\kappa+1}\ge\sqrt x-0.01$ for sufficiently large $\kappa$.
Therefore, $\frac{x}{t^2}\in[0,2]$, and hence for every
$t \in [\min(t_i,\widetilde t_i),\,\max(t_i,\widetilde t_i)]$ we have
\[
|F'(t)|\le \frac12.
\]
By the mean value theorem,
\[
|F(\widetilde t_i)-F(t_i)|\le \frac12\,|\widetilde t_i-t_i|
\le 2^{-\kappa}.
\]
Consequently,
\[
|\widetilde t_{i+1}-t_{i+1}|
\le |\widetilde t_{i+1}-F(\widetilde t_i)| + |F(\widetilde t_i)-F(t_i)|
\le 2^{-\kappa}+2^{-\kappa}
= 2^{-\kappa+1}.
\]

\end{enumerate}
\end{proof}

\subsection[\texorpdfstring{Proof of \Cref{theorem:securityofweakprfsfromlwe}}{Security of the LWE-based WPRF}]{Proof of \Cref{theorem:securityofweakprfsfromlwe}}

\label{proofoftheorem:securityofweakprfsfromlwe}

\begin{theorem}[Restating \Cref{theorem:securityofweakprfsfromlwe}]
Let all parameters be as in~\Cref{construct:weakprffromlwe}.
For any $m,t\in\N$, if the $(n,m,q,\alpha,t,\epsilon)$-$\LWE$ assumption holds, then $\ff G_\lambda$ in \Cref{construct:weakprffromlwe} is an $(m,t-m\cdot\poly(\lambda),\epsilon')$-weak pseudorandom function family, where
$$
\epsilon' =  \tau \cdot \kappa^{\tau} \cdot \epsilon + m\cdot \left(\left(\kappa^{\tau}\cdot \frac{1}{\lambda}\right)^{\tau-1} +  \tau \cdot 2^{-\Omega(\kappa)} \right)
$$
The concrete bound applies whenever the displayed circuit-size bound is positive.
In particular, if the $\LWE$ assumption holds, then $\ffe G$ is a secure WPRF.
\end{theorem}

\begin{proof}
To aid the proof we first define a sequence of auxiliary (randomized) functions $\tilde{ G}^{(i^*)}$ for $0\leq i^* \leq \tau$. The function $\tilde{G}^{(i^*)}$ is defined identically as the WPRF  $G$ except that it samples $\vcy^{(i^*)} \sample  \Z_q^{\kappa^{(i^*)}}$ and starts at the $i^* + 1$-th round. Formally, on input $\vc a \in \Z_q^{n}$,
\begin{enumerate}
        \item $\vc z^{(i)} \gets \mat S^{(i)} \cdot \vc a + \vc p^{(i)} \in \Z_q^{\kappa^{(i)}}$ for $i \in [\tau]$ in parallel.
        \item $\vc y^{(i^*)} \sample \Z_q^{\kappa^{(i^*)}}$.
        \item For $i = i^*+1, \dots, \tau$ sequentially:
        \begin{enumerate}
            \item $r^{(i)} \gets \lceil \vc y^{(i-1)}  \rfloor_2 \in \bina^{\kappa^{(i-1)}}$
            \item $\vc e^{(i)} \gets S^{(B)}_{\chi}(r^{(i)}) \in \Z_q^{\kappa^{(i)}}$
            \item $\vc y^{(i)} \gets \vc z^{(i)}  + \vc e^{(i)}  \in \Z_q^{\kappa^{(i)}}$
        \end{enumerate}
        \item Output $y \gets \vc y^{(\tau)} \in \Z_q$.
    \end{enumerate}

Fix any $m \in \N$. We proceed via a sequence of hybrids.

\textbf{Hybrid $H_0$.}  Sample $m$ random inputs $\vc x= (x_1,\ldots,x_m) \sample \Set X^m$, a random key $k \sample \Set K_{\lambda}$ for $\ff G_{\lambda}$, and give the adversary $(\vc x, G_k(\vc x))$.

\textbf{Hybrid $H_{1, i^*}$} for $0 \leq i^* \leq \tau$.  Sample $m$ random inputs $\vc x= (x_1,\ldots,x_m) \sample \Set X^m$, a random key $k \sample \Set K_{\lambda}$ for $\ff G_{\lambda}$, and give the adversary $(\vc x, \tilde{G}^{(i^*)}_k(\vc x))$.

\textbf{Hybrid $H_2$.} Sample $m$ random inputs $\vc{x}=(x_1,\ldots,x_m) \sample \Set X^m$, $m$ random outputs $\vc u = (u_1, \dots, u_m) \sample \Set Y^m$, and give the adversary $(\vc x, \vc u)$.

By definition, $H_{1, \tau}$ is the same as $H_2$.

\begin{lemma}
\label{lemma:weakprffromlweH0indis}
Games $\game H_{0}$ and $H_{1, 0}$ are $\left(0,\infty, m\cdot \left(\kappa^{\tau}\cdot \frac{1}{\lambda}\right)^{\tau-1}\right)$-\emph{indistinguishable}.
\end{lemma}
\begin{proof}

We define a coupling distribution $(X, Y)$ as follows.

\textbf{Coupling distribution $(X, Y)$}. Sample $m$ random inputs $\vc x= (x_1,\ldots,x_m) \sample \Set X^m$, a random key $k \sample \Set K_{\lambda}$ for $\ff G_{\lambda}$, and output $((\vc x,G_k(\vc x)),(\vc x,\tilde G^{(0)}_k(\vc x)))$, using fresh independent randomness for each evaluation of $\tilde G^{(0)}$.

By definition, $X \sim H_0$, $Y \sim H_{1, 0}$.
By \Cref{lem:couplingbound}, it holds that
\begin{equation*}
\begin{split}
\SD(H_0, H_{1, 0}) \leq& \Pr[X \neq Y]\\
        =& \Pr_{\substack{\vc x \sample \Set X^m \\ k\sample \Set K}}[\tilde{G}^{(0)}_k(\vc x) \neq G_k(\vc x)]\\
           \leq& m\cdot  \Pr_{\substack{x \sample \Set X \\ k\sample \Set K}}[\tilde{G}^{(0)}_k(x) \neq G_k(x)]\\
           \leq& m\cdot \left(\kappa^{\tau}\cdot \frac{1}{\lambda}\right)^{\tau-1}
\end{split}
\end{equation*}
where the second to last inequality is due to a union bound, and the last inequality is due to~\Cref{claim:h0h1collision}.

\begin{claim}
\label{claim:h0h1collision}
For any $x \in \Set X$, it holds that
$$
\Pr_{k\sample \Set K}[\tilde{G}^{(0)}_k(x) \neq G_k(x)] \leq  \left(\kappa^{\tau}\cdot \frac{1}{\lambda}\right)^{\tau-1}
$$
\end{claim}

\begin{proof}
Fix any $x \in \Set X$.
We consider the transcript of the intermediate results during the evaluation of $G_{\skey}(x)$ and $\tilde{G}^{(0)}_{\skey}(x)$.
We use $X$ to denote the random variable during the evaluation of $G_{\skey}(x)$ (e.g., $\vcy$, $\vcz$).
We use $\tilde{X}$ to denote the random variable during the evaluation of $\tilde{G}^{(0)}_{\skey}(x)$ (e.g.,~$\tilde{\vcy}$, $\tilde{\vcz}$).

Observe that if $\tilde{G}^{(0)}_k(x)  \neq G_k(x)$, then it must hold that for all $i \in [\tau]$, $r^{(i)} \neq \widetilde{r}^{(i)}$. This is due to the fact that each iteration is deterministic. For $i=2,\ldots,\tau$, write $A_{i-1}=\{r^{(i-1)}\ne\widetilde r^{(i-1)}\}$ and consider the probability of $r^{(i)}\ne\widetilde r^{(i)}$ conditioned on $A_{i-1}$:
\begin{equation}
\begin{split}
    \Pr_{k \sample \Set K}\left[r^{(i)} \neq \widetilde{r}^{(i)} \mid A_{i-1} \right]
    =&\Pr_{k \sample \Set K} \left[\lceil \vc y^{(i-1)}  \rfloor_2 \neq \lceil \widetilde{\vc y}^{(i-1)}  \rfloor_2  \mid A_{i-1} \right] \\
    =&\Pr_{k \sample \Set K} \left[\lceil \vc z^{(i-1)} + \vc e^{(i-1)}  \rfloor_2 \neq \lceil \vc z^{(i-1)} + \widetilde{\vc e}^{(i-1)}  \rfloor_2  \mid A_{i-1} \right] \\
    \leq& \sum_{j=1}^{\kappa^{(i-1)}}\Pr_{k\sample\Set K}\left[\lceil z_j^{(i-1)}+e_j^{(i-1)}\rfloor_2\neq\lceil z_j^{(i-1)}+\widetilde e_j^{(i-1)}\rfloor_2\mid A_{i-1}\right]\\
    \leq& \kappa^{\tau} \cdot \frac{1}{\lambda}.
\end{split}
\end{equation}
The first inequality follows from a union bound over all entries, and the second inequality follows from the fact that each entry of
$\vc z^{(i-1)} = \vc p^{(i-1)} + \star$ is uniformly distributed.
Conditioned on all previous-layer values, each coordinate $z_j^{(i-1)}$ remains uniform and independent of the two errors, because the fresh padding coordinate $p_j^{(i-1)}$ is uniform. For any two fixed integer errors $e,\widetilde e\in[-B,B]$, the most significant bits of $z+e\bmod q$ and $z+\widetilde e\bmod q$ differ on at most $2|e-\widetilde e|\le4B$ residues $z\in\mathbb Z_q$. Thus the disagreement probability is at most $4B/q\le1/\lambda$.

It therefore follows that
\begin{equation}
\begin{split}
   \Pr_{k\sample \Set K}[\tilde{G}^{(0)}_k(x) \neq G_k(x)]
    \leq&  \Pr_{k\sample \Set K}\left[\bigcap_{i=1}^{\tau} r^{(i)} \neq \widetilde{r}^{(i)} \right] \\
    \leq & \prod_{i=2}^{\tau}\Pr\left[r^{(i)} \neq \widetilde{r}^{(i)} \; |\;r^{(i - 1)} \neq \widetilde{r}^{(i - 1)}\right] \\
    \leq&  \left(\kappa^{\tau}\cdot \frac{1}{\lambda} \right)^{\tau-1}
\end{split}
\end{equation}
\end{proof}

\end{proof}

\begin{lemma}
\label{lemma:weakprffromlweH1indis}
For $i^*\in\{0,\ldots,\tau-1\}$, games $\game H_{1,i^*}$ and $\game H_{1,i^*+1}$ are
\[
\left(0,t-m\cdot\poly(\lambda),\ \kappa^\tau\epsilon+m\kappa^\tau2^{-\Omega(\kappa)}\right)\text{-indistinguishable},
\]
whenever the stated circuit-size bound is positive.
\end{lemma}
\begin{proof}
Suppose that a circuit $D$ of size $t-m\cdot\poly(\lambda)$ distinguishes $H_{1,i^*}$ and $H_{1,i^*+1}$ with advantage $\epsilon'$, where the polynomial is chosen large enough to cover the simulations below. We build a circuit $D'$ distinguishing the following games with advantage at least $\epsilon'-m\kappa^\tau2^{-\Omega(\kappa)}$. Its simulation overhead, including the subsequent hybrid reduction over secrets, is at most $m\cdot\poly(\lambda)$. Thus the $(n,m,q,\alpha,t,\epsilon)$-$\LWE$ assumption bounds this advantage by $\kappa^\tau\epsilon$.
\begin{itemize}
    \item Game $\game G_0$: Sample $\mat A \sample \mathbb{Z}_q^{n\times m}$, $\mat S = (\vc s_1, \dots, \vc s_{\kappa^{\tau}}) \sample \mathbb{Z}_q^{\kappa^{\tau} \times n}$, and $\mat E \sample \chi^{\kappa^{\tau} \times m}$. Give the adversary $(\mat A, \mat S \cdot \mat A + \mat E)$.
    \item Game $\game G_1$: Sample $\mat A \sample \mathbb{Z}_q^{n\times m}$, and $\mat U \sample \mathbb{Z}_q^{\kappa^{\tau} \times m}$. Give the adversary $(\mat A, \mat U)$.
\end{itemize}
The distinguisher $D'$ obtains $(\mat A, \mat B) \in \Z_q^{n\times m} \times \Z_q^{\kappa^{\tau} \times m}$ as input,
it parses the columns of $\mat A$ as $\vc x = (\mat A_1, \dots, \mat A_m)$, where $\mat A_i \in \Z_q^n$ is the $i$-th column of $\mat A$. It samples a random key $k \sample \Set K_{\lambda}$,
samples $\vc p \sample \Z_q^{\kappa^{(i^* + 1})}$,
and gives the adversary $(\vc x, \tilde{G}^{(i^* + 1)}_k(\vc x, \mat B))$, where $\tilde{G}^{(i^* + 1)}_k(\vc x, \mat B)$ is defined exactly as $\tilde{G}^{(i^* + 1)}_k(\vc x)$, except that when evaluating the $j$-th input for $j\in [m]$, instead of sampling $\vc y^{(i^* + 1)} \sample \Z_q^{\kappa^{(i^* + 1)}}$, it sets $\vc y^{(i^* + 1)}$ to be the first $\kappa^{(i^*+1)}$ entries of the $j$-th column of $\mat B$, added with $\vc p$.
\begin{itemize}
    \item If $D'$ receives $(\mat A, \mat B)$ from Game $\game G_1$, then $D'$ perfectly simulates $H_{1, i^* + 1}$.
    \item If $D'$ receives $(\mat A, \mat B)$ from Game $\game G_0$, then $D'$ simulates $H_{1, i^*}$ except with a statistical distance of at most  $$m \cdot \kappa^{\tau}\cdot \SD(\chi, S^{(B)}_{\chi}(U_{\kappa})) =  m\cdot \kappa^{\tau} \cdot 2^{-\Omega(\kappa)}.$$
\end{itemize}
Therefore, $D'$ has the distinguishing advantage of at least $\epsilon' -  m\cdot \kappa^{\tau} \cdot 2^{-\Omega(\kappa)}$.

\end{proof}
The theorem follows from combining \Cref{lemma:weakprffromlweH0indis} and \Cref{lemma:weakprffromlweH1indis}. Since $\kappa^\tau=2^{o(\kappa)}$, polynomial factors in $\kappa^\tau$ are absorbed into the $2^{-\Omega(\kappa)}$ term; this also covers the statistical error from using $S^{(B)}_\chi$ to simulate error samples in the reductions.

\end{proof}

\bibliographystyle{alpha}
\bibliography{abbrev,refs}

\end{document}